\documentclass[journal,onecolumn, 12pt, draftcls]{IEEEtran}

\usepackage[font=small]{caption}
\usepackage[labelsep=period]{caption}
\usepackage{float}
\usepackage[latin1]{inputenc}
\usepackage{graphicx}
\usepackage[cmex10]{amsmath}
\usepackage[margin=2cm,nohead]{geometry}
\usepackage{cite}
\usepackage{amsfonts}
\usepackage{amsmath}
\usepackage{amsthm}
\usepackage{mathrsfs,dsfont}
\usepackage{amssymb,mathrsfs}
\usepackage[mathscr]{euscript}
\usepackage{subcaption}
\usepackage{epstopdf} 
\usepackage{color}
\usepackage{pgfplots}
\pgfplotsset{compat=newest}
\usetikzlibrary{plotmarks}
\usepackage{grffile}
\usepackage{amsmath}

\newtheorem{theorem}{Theorem}
\newtheorem{corollary}{Corollary}

\newtheorem{lemma}{Lemma}

\def\Sdc{\mathcal{S}^\text{dc}}

\def\Pout{\mathcal{P}_\text{out}}
\def\Poutdc{\mathcal{P}_\text{out}^{\,\text{dc}}}
\def\Poutrc{\mathcal{P}_\text{out}^{\,\text{rc}}}
\def\Poutcc{\mathcal{P}_\text{out}^{\,\text{cc}}}

\def\a{\alpha}
\def\PU{P_\text{U}}
\def\A{\mathcal{A}}

\def\C{\mathcal{C}}

\def\KUX{K_\text{UX}}
\def\KRD{K_\text{RD}}
\def\GUX{\Gamma_\text{UX}}
\def\GUD{\Gamma_\text{UD}}
\def\PU{P_\text{U}}
\def\PR{P_\text{R}}
\def\lUX{\ell_\text{UX}}
\def\aUX{\alpha^{}_\text{\tiny{UX}}}
\def\aRD{\alpha^{}_\text{\tiny{RD}}}
\def\aUD{\alpha^{}_\text{\tiny{UD}}}
\def\OUX{\Omega_\text{UX}}
\def\OUD{\Omega_\text{UD}}
\def\tX{\theta_\text{X}}
\def\tD{\theta_\text{D}}
\def\tR{\theta_\text{R}}
\def\tC{\theta_\mathcal{C}}
\def\lC{\ell_\mathcal{C}}
\def\GRD{\Gamma_{{\text{R}}{\text{D}}}}
\def\lRD{\ell_{\text{RD}}}
\def\lUR{\ell_{\text{UR}}}
\def\ORD{\Omega_{{\text{R}}{\text{D}}}}
\def\rD{r^{}_\text{D}}
\def\lUD{\ell_{\text{UD}}}

\def\rdc{r_\mathcal{C}^\text{dc}}
\def\rrc{r_\mathcal{C}^\text{rc}}

\def\rCdc{\tilde{r}_\mathcal{C}^\text{dc}}
\def\rCrc{\tilde{r}_\mathcal{C}^\text{rc}}
\def\rCcc{\tilde{r}_\mathcal{C}^\text{cc}}
\def\hOPTdc{\tilde{h}_\mathcal{C}^\text{dc}}

\def\xc{x_\mathcal{C}}
\def\yc{y_\mathcal{C}}
\def\tCopt{\tilde{\theta}_\mathcal{C}}
\def\tCoptdc{\tilde{\theta}_\mathcal{C}^{\,\text{dc}}}

\def\PLoS{\mathcal{P}_\text{LoS}}

\def\hdc{\tilde{h}_\text{D}^{\text{dc}}}
\def\hrc{\tilde{h}_\text{D}^{\text{rc}}}
\def\hcc{\tilde{h}_\text{D}^{\text{cc}}}
\def\tetdc{\tilde{\theta}_\text{D}^{\,\text{dc}}}
\def\tetrc{\tilde{\theta}_\text{D}^{\,\text{rc}}}
\def\tetcc{\tilde{\theta}_\text{D}^{\,\text{cc}}}

\def\rR{r^{}_\text{R}}
\def\fiR{\varphi^{}_\text{R}}
\def\rOrc{\tilde{\rho}_\mathcal{C}^{\,\text{rc}}}
\def\rOcc{\tilde{\rho}_\mathcal{C}^{\,\text{cc}}}
\def\PoutY{\mathcal{P}_\text{out}^\text{Y}}
\def\hY{\tilde{h}_\text{D}^\text{Y}}
\def\rY{r_\mathcal{C}^\text{Y}}
\def\rCY{\tilde{r}_\mathcal{C}^\text{Y}}
\def\hOPTY{\tilde{h}_\mathcal{C}^\text{Y}}
\def\SY{\mathcal{S}^\text{Y}}

\begin{document}
\title{Ultra Reliable UAV Communication Using Altitude and Cooperation Diversity}
\author{Mohammad Mahdi Azari, Fernando Rosas, Kwang-Cheng Chen, and Sofie Pollin}
\maketitle

\begin{abstract}
The use of unmanned aerial vehicles (UAVs) that serve as aerial base stations is expected to become predominant in the next decade. However, in order for this technology to unfold its full potential it is necessary to develop a fundamental understanding of the distinctive features of air-to-ground (A2G) links. As a contribution in this direction, this paper proposes a generic framework for the analysis and optimization of the A2G systems. In contrast to the existing literature, this framework incorporates both height-dependent path loss exponent and small-scale fading, and unifies a widely used ground-to-ground channel model with that of A2G for analysis of large-scale wireless networks. We derive analytical expressions for the optimal UAV height that minimizes the outage probability of a given A2G link. Moreover, our framework allows us to derive a height-dependent closed-form expression and a tight lower bound for the outage probability of an \textit{A2G cooperative communication} network. Our results suggest that the optimal location of the UAVs with respect to the ground nodes does not change by the inclusion of ground relays. This enables interesting insights in the deployment of future A2G networks, as the system reliability could be adjusted dynamically by adding relaying nodes without requiring changes in the position of the corresponding UAVs.
\end{abstract}



\section{Introduction}

Aerial telecommunication platforms have been increasingly used as an innovative method to enable robust and reliable communication networks. Facebook \cite{Facebook} and Google \cite{katikala2014google} have been planning to establish massive networks of unmanned aerial vehicles (UAVs) to provide broadband connectivity to remote areas. Amazon has announced a research plan to explore new means of delivery service making use of small UAVs to shorten the delivery time \cite{amazon}. The ABSOLUTE project in Europe aims to provide a low latency and large coverage networks using aerial base stations for capacity enhancements and public safety during temporary events \cite{absolute}. UAVs have also been used in the context of Internet of Things (IoT) to assist low power transmitters to send their data to a destination \cite{lien2011toward,dhillon2015wide}. Moreover, it has been shown that UAVs as low altitude platforms (LAPs) can be integrated into a cellular network to compensate cell overload or site outage \cite{rohde2012interference}, to enhance public safety in the failure of the base stations \cite{merwaday2015uav}, and to boost the capacity of the network.

To enable the deployment of such systems, it is important to model the reliability and coverage of the aerial platforms. In particular, recent studies have shown that the location and height of a UAV can significantly affect the air-to-ground (A2G) link reliability. In \cite{li2015drone} the optimal placement of UAVs, acting as relays, is studied without considering the impact of altitude on its coverage range. This issue is addressed in recent efforts \cite{al2014optimal,bor2016efficient,mozaffari2015unmanned}, where the authors show the dependency of the network performance on the altitude of UAVs. In \cite{al2014optimal} the impact of altitude on the coverage range of UAVs is studied, without providing a closed-form expression showing the dependency of the coverage radius to the system parameters. This work is based on a simple disk model, where the path loss is compared with a given threshold value. The authors in \cite{bor2016efficient} use a model similar to the one reported in \cite{al2014optimal} to investigate an optimum placement of multiple UAVs to cover the maximum possible users by applying a numerical search algorithms. In \cite{mozaffari2015unmanned} the average sum-rate is studied as a function of UAV altitude over a known coverage region. All these works, however, ignore the stochastic effects of multipath fading, which is an essential feature of an A2G communication link \cite{matolak2012air}. Moreover, the path loss modeling is based on free space conditions \cite{al2014modeling}, making it impossible to address the effects of low altitudes which is of practical importance due to the regulation constraints.

Besides the link characterization, the next step to further improve the reliability of a UAV link is to study the cooperation of the UAV with the existing terrestrial network. Cooperative communication techniques are known to significantly enhance the reliability of a wireless system. Three main cooperation protocols named amplify-and-forward (AF), decode-and-forward (DF) and coded cooperation (CC) together with outage performance have been studied in \cite{laneman2004cooperative,wang2011general,yu2012outage,behnad2013generalized,eddaghel2013outage,azari2014probabilistic}, relying on identical propagation and fading models for the first and second hops of the cooperative communication link. In \cite{azari2014probabilistic} the results show that CC and DF outperform AF in terms of outage performance, however more complexity is imposed to the system. In fact, DF has a moderate complexity and outage performance among these strategies and discards the disadvantage of AF where the noise is amplified in relays. All of these studies focus on terrestrial cooperative networks, however in this paper we incorporate ground-to-ground (G2G) and A2G communication links in an \textit{A2G cooperative communication} system. To the best of our knowledge, this is the first paper investigating the outage performance of such network as function of UAV altitude. This allows to quantify the benefits of ground relaying nodes at every altitude in A2G networks and enables to measure the performance gains when aerial and ground nodes cooperate. However, a fair and accurate analysis of an A2G cooperative network would require a generic framework which is valid for both G2G and A2G communication links, enabling to consider a continuum of different propagation conditions.

To address this need, we propose a generic framework that extends a widely used model for G2G wireless links towards A2G channels. Moreover, this framework considers an altitude-dependent path loss exponent and fading function, supplementing the existent A2G channel models. Using this analytical framework, we also significantly extend and refine our previous works \cite{glob,azari2016coverage,wi-uav}. In fact, the resutls presented in our initial study \cite{glob} are restricted to consider the same path loss exponent at different altitudes, which is an unrealistic assumption considering the different propagation environments. Moreover, we derive the height of the UAV for minimum outage probability at every user location and the optimum UAV altitude resulting in maximum coverage radius, which are not presented in \cite{azari2016coverage}.  


We also extend our previous work by analyzing an A2G cooperative network where each communication link posses different channel statistics. To this end, we analyze the height-dependent outage performance of the network when adopting a DF protocol at the ground relays, deriving a lower bound for the end-to-end outage probability. We quantify the benefits of the cooperative relaying, showing that the reliability of A2G communication strongly benefits from the introduction of ground relays. Furthermore, we show that the optimal transmission height is not much affected by the existence of ground relays, which allows to include them dynamically without loosing the optimality of the UAV location.

The rest of this paper is organized as follows. In Section \ref{system model} we present the system model. The problem is stated in Section \ref{problem statement}. The system outage performance and the optimum altitude in direct A2G communication links are derived in Section \ref{direct communication} followed by Section \ref{relaying communication} where the outage performance of the relaying strategy and the corresponding lower bound are analyzed. Our results are discussed and confirmed in Section \ref{numerical}. Finally, the main conclusions are presented in Section \ref{conclusion}. 



\section{System Model} \label{system model}

We consider a cellular system where a UAV provides wireless access to terrestrial mobile devices, serving as an aerial base station. The UAV is placed in an adjustable altitude $h$, aiming to communicate with a ground node D either directly or through a terrestrial relay R, as illustrated in Figure \ref{system}. The elevation angles of the UAV with regard to D and R are denoted as $\theta_\text{D}$ and $\theta_\text{R}$, respectively. The relaying nodes are randomly distributed, following a Poison point process (PPP) with a fixed density $\lambda$ in a disk $\C$ centered at the projection of the UAV on the ground, denoted as O. A polar coordinate system with the origin at O is considered, so that the location of the nodes D and R can be described by $(r_\text{D},\varphi_\text{D})$ and $(r_\text{R},\varphi_\text{R})$ respectively. We note that, in fact, in our model we could also consider aerial relays, but we believe ground relays are more meaningful from a practical perspective. 

In the sequel, Section \ref{channel modeling} describes the channel model and the corresponding characteristics of fading and path loss in terms of altitude are discussed in Section \ref{rician factor modeling}.  

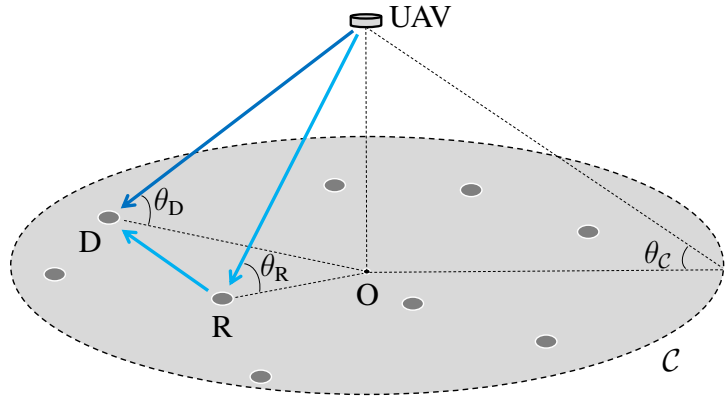
\begin{figure}[t!]
  \centering
  \begin{tikzpicture}
  \node[above right] (img) at (0,0) {\includegraphics[width=0.6\columnwidth]{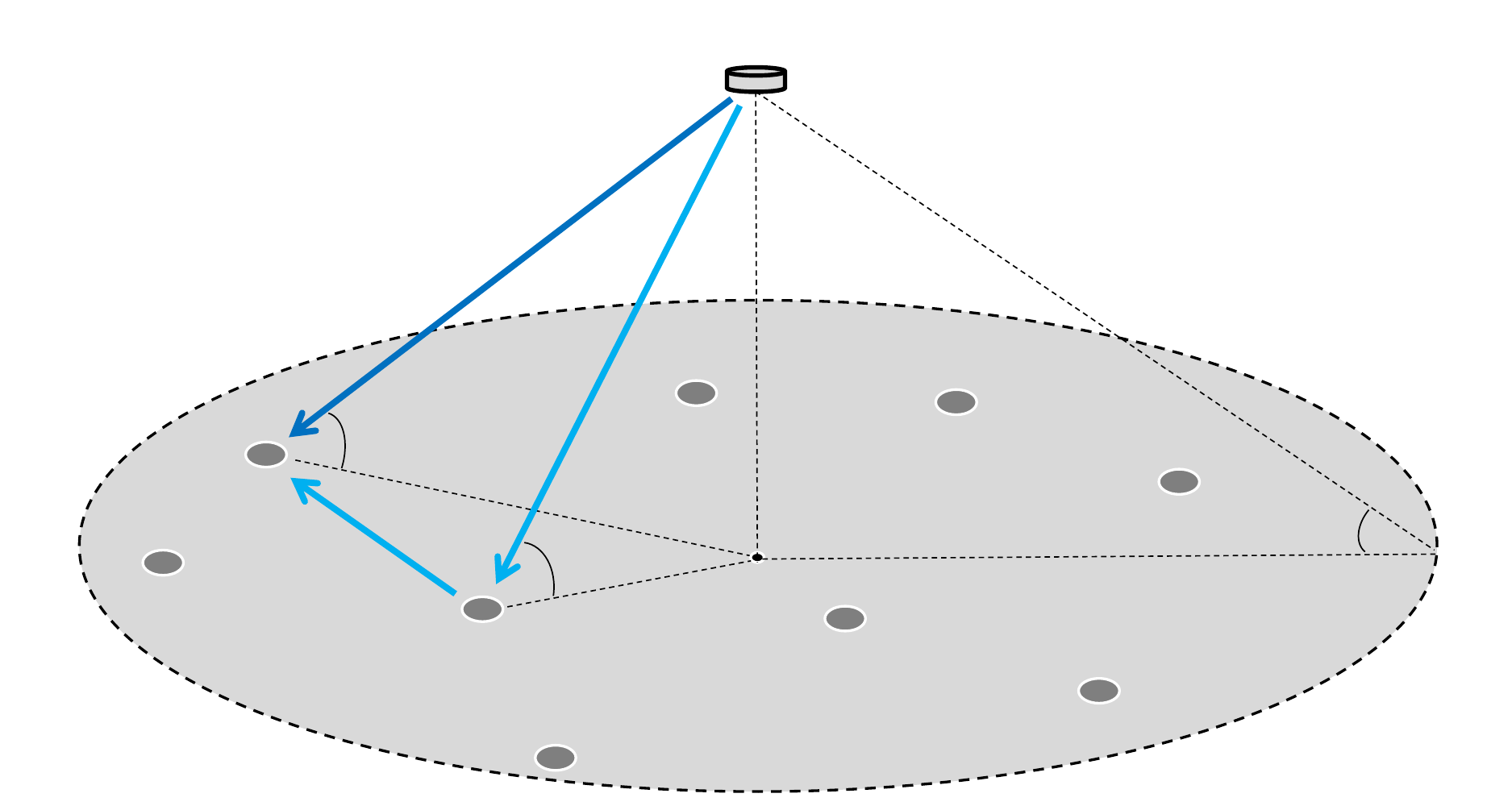}};
  \node at (155pt,45pt) {$\text{O}$};
  \node at (175pt,150pt) {$\text{UAV}$};
  \node at (50pt,65pt) {$\text{D}$};
  \node at (100pt,33pt) {$\text{R}$};
  \node at (120pt,55pt) {$\tR$};
  \node at (80pt,80pt) {$\tD$};
  \node at (270pt,20pt) {$\C$};
  \node at (265pt,60pt) {$\tC$};
  \end{tikzpicture}
  \caption{A typical air-to-ground (A2G) wireless networking using a UAV in presence of randomly distributed ground relays over the coverage region $\C$. In our analysis, we consider both a ground destination (D), as well as cooperative relays (R).}\label{system}
\end{figure}


\subsection{Channel Modeling} \label{channel modeling}

The wireless channel between any pair of nodes is assumed to experience small-scale fading and large-scale path loss. Therefore, the instantaneous SNR between the UAV, denoted as U, and a ground receiver X, which is either R or D, can be modeled as
\begin{equation}\label{GUX}
\GUX = \frac{A \PU}{N_0 \lUX^{~\aUX}}~\OUX;~~~\text{X} \in \{\text{D},\text{R}\},
\end{equation} 
where $\PU$ is the UAV's transmit power, $N_0$ is the noise power, $\lUX$ is the distance between the UAV and the node X, $\aUX$ is the path loss exponent, $A$ is a constant which depends on the system parameters such as operating frequency and antenna gain, and $\OUX$ is the fading power where $\overline{\Omega}_\text{UX} = 1$. 

In order to model the small-scale fading between the UAV and any ground node X, a Rician distribution is an adequate choice due to the possible combination of LoS and multipath scatterers that can be experienced at the receiver \cite{kandeepan2014aerial,kandeepan2011energy,matolak2012air}. Using this model, $\Omega_\text{UX}$ adopts a non-central chi-square probability distribution function (PDF) expressed as \cite{simon2005digital}
\begin{equation}\label{OUX}
f_{\OUX}(\omega) =
\frac{(\KUX+1)e^{-\KUX}}{\overline{\Omega}_\text{UX}}~e^\frac{-(\KUX+1)\omega}{\overline{\Omega}_\text{UX}}~
I_0\left(2\sqrt{\frac{\KUX(\KUX+1)\omega}{{\overline{\Omega}_\text{UX}}}}\right);~~~\omega\geq0.
\end{equation}         
Above $I_0(\cdot)$ is the zero-order modified Bessel function of the first kind, and $\KUX$ is the Rician factor defined as the ratio of the power in the LoS component to the power in the non-LoS multipath scatters. In this representation, $\KUX$ reflects the severity of the fading. In effect, if $\KUX = 0$, the equation \eqref{OUX} is reduced to an exponential distribution indicating a Rayleigh fading channel, while if $\KUX \rightarrow \infty$ the channel converges to an AWGN channel. Accordingly, more severe fading conditions correspond to a channel with lower $\KUX$.

Following a similar rationale, the instantaneous SNR between an arbitrary relaying node $\text{R}$ as the transmitter and the receiver D can be modeled as
\begin{equation}\label{G_G1G2}
\GRD = \frac{A \PR}{N_0 \lRD^{~\aRD}}~\ORD,
\end{equation} 
where $\PR$ is the relay's transmit power which is assumed to be the same for all the relaying nodes, $\lRD$ is the distance between R and D, $\aRD$ is the G2G path loss exponent between R and D, and $\ORD$ is also modeled by a Rician distribution with $\overline{\Omega}_{{\text{R}}{\text{D}}} = 1$, where its PDF can be written using \eqref{OUX} but considering a different Rician factor denoted as $\KRD$. Finally, it is assumed that the fading statistics between different pair of nodes are independent.


\subsection{Rician Factor and Path Loss Exponent Modeling} \label{rician factor modeling}

Intuitively, the height of the UAV affects the propagation characteristics of the A2G communication link since the LoS condition and the environment between U and X alter as $\tX$ varies. It has been shown that the elevation angle of the UAV with respect to the ground node plays a dominant role in determining the Ricean factor \cite{hagenauer1987maritime,shimamoto2006channel}. Consequently, we model the Ricean factor as a function of the elevation angle by introducing the non-decreasing function $\KUX = K(\tX)$. Indeed, a larger $\tX$ implies a higher LoS contribution and less multipath scatters at the receiver resulting in a larger $\KUX$. Thus, G2G communication ($\tX = 0$) experience the most severe multipath conditions and hence $\KUX$ takes its minimal value $\kappa_0 = K(0)$, whereas at $\tX = \pi/2$ it adopts the maximum value $\kappa_{{\pi}/{2}} = K({\pi}/{2})$. Note that $\KRD = \kappa_0$, as it corresponds to a G2G link.

Following a similar rationale, the path loss is also influenced by the elevation angle such that $\aUX$ might decrease as the UAV's elevation angle $\tX$ increases. In this way, G2G links, where $\tX = 0$, endure the largest $\aUX$ while at $\tX = \pi/2$ the value of $\aUX$ is the smallest. Therefore, we model the path loss exponent dependency on the elevation angle by introducing a non-increasing function of $\aUX = \a(\tX)$ and define the shorthand notations $\a_0 = \a(0)$ and $\a_{\pi/2} = \a(\pi/2)$. Accordingly, G2G links have $\aRD = \a_0$. 

The analysis presented in the following sections is based on a general dependency of $K(\tX)$ and $\a(\tX)$ on $\tX$, in order to provide comprehensive results which are valid over a variety of conceivable scenarios. Therefore, in any concrete scenario the results can be instantiated by determining the estimated functional form for $K(\tX)$ and $\a(\tX)$. However, in Section \ref{numerical} we adopt a particular parameterized family of functions in order to illustrate our results.




\section{Problem Statement} \label{problem statement}

An A2G channel benefits from a lower path loss exponent and lighter small-scale fading compared to a G2G link, while having a longer link length which deteriorates the received SNR. Interestingly, these two opposite effects can be balanced by optimizing the UAV height. Thus, our fundamental concern is to find the best position of a UAV for optimizing the link reliability, and to study if such an optimized positioning can have a positive impact on the UAV coverage area. Finally, we are also interested in studying the effect of ground relays --as a means of reliability enhancement and coverage extension-- on the optimal UAV altitude.


Link reliability is usually evaluated using the outage probability, which is defined as 
\begin{equation} \label{outage definition Y}
\PoutY \triangleq \mathds{P}(\Gamma \leq \xi),
\end{equation}
where $\mathds{P}(\text{E})$ indicates the probability of an event E, $\Gamma$ is the instantaneous SNR at the receiver, $\xi$ is the SNR threshold which depends on the sensitivity of the receiver, and $\text{Y} \in \{\text{dc},\text{rc},\text{cc}\}$ indicates the strategy employed for communication which is an abbreviation for direct communication, relaying communication and cooperative communication, respectively. 

The A2G channel characteristic and hence the received SNR at the ground destination D is dependent on the relative position of the UAV to D determined by $h$ and $r_\text{D}$, and hence the outage probability can be written as $\PoutY = \PoutY(\rD,h)$. At a given $\rD$, the optimum altitude of the UAV $\hY$ for maximum reliable link is defined as
\begin{equation}\label{h*Y}
\hY = \arg\min_{h\in[0,\infty)} \PoutY(\rD,h);~~\text{Y} \in \{\text{dc},\text{rc},\text{cc}\}.
\end{equation}

On the other hand, for a given altitude $h$ the radius of UAV's coverage area $\C$ is defined as the maximum distance $\rD$ within which the outage probability remains below or equals to a target $\varepsilon$. For a larger $\rD$ the outage probability $\PoutY(\rD,h)$ is higher due to the larger path loss and more severe fading. Thus, the boundary of the coverage region is characterized by
\begin{equation}\label{PoutY=eps}
\PoutY(\rD,h) = \varepsilon.
\end{equation}
For a given $h$, we denote the radius $\rD$ satisfying the above equation as $\rY$. The set of all pairs $(\rY,h)$ that meet the above equation constitute a system \textit{configuration space} denoted as $\SY$. We intend to obtain the maximum coverage radius $\rCY$ in $\SY$, by locating the UAV at the optimum altitude $\hOPTY$. To this end, the problem can be formulated as
\begin{align}\label{rCY} \nonumber
&\rCY = \max_{h\in[0,\infty)} \rY,  \\
&s.t.~~~(\rY,h) \in \SY
\end{align}
and $\hOPTY$ is the altitude that $(\rCY,\hOPTY) \in \SY$, i.e. the optimal altitude at which the coverage radius is maximized.

Note that increasing $\rD$ leads to an increase in the corresponding minimum outage probability $\PoutY(\rD,\hY)$ (incurred at altitude $\hY$) due to the adverse effect of the larger link length. Accordingly, at the border of coverage region, i.e. $\rD = \rCY$, $\PoutY(\rD,\hY)$ reaches the predefined target outage performance $\varepsilon$. In this manner, the optimization problem in \eqref{rCY} is related to \eqref{h*Y}, but we adopt different approaches to solve them. 


\section{Direct Air-to-Ground Communication} \label{direct communication}

This section we analyze the outage probability of transmissions going directly from the UAV to the destination. The outage probability is studied as function of UAV altitude in Section \ref{DirectOutProb} and then the results are used to investigate the optimal placement of the UAV for maximum reliability and range of communication in Section \ref{DirectMaxCov}.


\subsection{Height-Dependent Outage Probability} \label{DirectOutProb}

Following \eqref{outage definition Y} the direct communication outage probability of the UAV--D link $\Poutdc$ is defined as $\Poutdc = \mathds{P}(\GUD \leq \xi)$. By using \eqref{GUX} and \eqref{OUX}, the outage probability can be rewritten as
\begin{align} \nonumber
\Poutdc&(\rD,h) = \mathds{P}\left(\frac{A \PU}{N_0 \lUD^{~\aUD}}~\OUD \leq \xi \right) \\ \label{Pout=1-Q}
 &= 1-Q\left(\sqrt{2 K(\tD)},\sqrt{2\xi~[1+K(\tD)]~{\lUD}^{\a(\tD)}/\gamma_\text{U}}\right),
\end{align}
where $\lUD = \sqrt{r^2_\text{D}+h^2}$, $\tD = \tan^{-1}(h/\rD)$, $Q(\cdot,\cdot)$ is the first order Marcum Q--function, and $\gamma_\text{U}$ is a shorthand notation for
\begin{equation} \label{LinkBudget}
\gamma_\text{U} \triangleq \frac{A \PU}{N_0}.
\end{equation}

We propose the following theorem to solve the optimization in \eqref{h*Y}, where the position of the UAV for minimum outage probability at every $\rD$ is obtained.

\begin{theorem} \label{t-hat}
For a given $\rD$ the optimal UAV altitude $\hdc$, as defined in \eqref{h*Y}, is given by
\begin{equation}
\hdc = \rD \cdot \tan(\tetdc),
\end{equation}
where $\tetdc$ is approximately obtained from
\begin{equation}\label{tD*Eq}
\sqrt{\frac{\xi}{\gamma_\text{U}}\left[\frac{\rD}{\cos(\tD)}\right]^{\a(\tD)}} \left[\frac{K'(\tD)}{K(\tD)}+\a'(\tD)\ln\left(\frac{\rD}{\cos(\tD)}\right)+\a(\tD)\tan(\tD) \right] =  \frac{K'(\tD)}{K(\tD)}.
\end{equation} 
\end{theorem}

\begin{proof}
The proof is given in Appendix \ref{t-hat proof}.
\end{proof}

Note that \eqref{tD*Eq} shows that $\tetdc$ is dependent on $\rD$ and hence the optimal altitude $\hdc = \rD \cdot \tan(\tetdc)$ is not generally a linear function of $\rD$. In fact, at short distances of $\rD$ an increase in the elevation angle $\tD$ could be more beneficial than that of large distances since the increase in the link length is smaller while the impact on the path loss exponent and the Rician factor is the same for any $\rD$. Therefore \eqref{tD*Eq} leads to a larger value of $\tetdc$ for an smaller $\rD$. 

In the next subsection we obtain the maximum coverage radius of the UAV $\rCdc$ and the corresponding optimum altitude $\hOPTdc$.


\subsection{Maximum Coverage Area} \label{DirectMaxCov}

The implicit relationship between $h$ and $\rdc$ in \eqref{PoutY=eps} can be rewritten using \eqref{Pout=1-Q} as
\begin{equation} \label{Q=1-eps}
Q\left(\sqrt{2 K(\tC)},\sqrt{2\xi~[1+K(\tC)]~{\lC}^{\a(\tC)}/\gamma_\text{U}}\right) = 1-\varepsilon,
\end{equation}
where $\lC = \sqrt{(\rdc)^2+h^2}$, $\tC = \tan^{-1}(h/\rdc)$. In order to find $\rCdc$ and $\hOPTdc$ from \eqref{rCY}, first we determine the configuration space using the following theorem.

\begin{theorem}
The configuration space $\Sdc$ is a one-dimensional curve in the $\rD$--$h$ plane, which is formed by all $(\rdc,h)$ obtained from
\begin{subequations}\label{h,rC}
\begin{align} \label{h}
h &= \Lambda(\tC) \cdot \sin(\tC), \\ \label{rC}
\rdc &= \Lambda(\tC) \cdot \cos(\tC),
\end{align}
where
\begin{align} \label{Lambda}
\Lambda(\tC) &= \left[\frac{\gamma_\text{U}\,\yc^2}{\xi\,(2+\xc^2)}\right]^{\frac{1}{\a(\tC)}}, \\ \label{x,y,tC}
\xc = \sqrt{2 K(\tC)},~\yc &= Q^{-1}(\xc,1-\varepsilon),~\tC \in [0,\frac{\pi}{2}].
\end{align}
\end{subequations}
Above, $Q^{-1}(\xc,\cdot)$ indicates the inverse Marcum Q--function with respect to its second argument.
\end{theorem}

\begin{proof}
Using the following auxiliary variables
\begin{subequations}\label{x,y-tC}
\begin{align}\label{x-tC}
\xc &= \sqrt{2 K(\tC)}, \\ \label{y-tC}
\yc &= \sqrt{2\xi~[1+K(\tC)]~{\lC}^{\a(\tC)}/\gamma_\text{U}},
\end{align}
\end{subequations}
the equation in \eqref{Q=1-eps} can be rewritten as
\begin{equation}
Q(\xc,\yc) = 1-\varepsilon,
\end{equation}
or equivalently
\begin{equation}\label{yc = Qinv}
\yc = Q^{-1}(\xc,1-\varepsilon).
\end{equation}
From \eqref{x,y-tC} and \eqref{yc = Qinv} one can write
\begin{equation} \label{lC=}
  \lC = \left[\frac{\gamma_\text{U}\,[Q^{-1}(\xc,1-\varepsilon)]^2}{\xi\,(2+\xc^2)}\right]^{\frac{1}{\a(\tC)}} \triangleq \Lambda(\tC).
\end{equation}
By using \eqref{lC=} in $h = \lC \cdot \sin(\tC)$ and $\rdc = \lC \cdot \cos(\tC)$ the desired result is attained.
\end{proof}

Note that all elements of $\Sdc$ can be described by exploring different values of $\tC$. In fact, when $\tC = 0$, one obtains $h = 0$ and $\rdc = \Lambda(0)$, while $h$ grows with $\tC$ reaching its maximum $h = \Lambda(\pi/2)$ when $\tC = \pi/2$ and $\rdc = 0$. The shape of the configuration space in $\rD$--$h$ plane is illustrated in Figure \ref{Sdc}, where $\Lambda(\tC)$ and $\tC$ are the radius and angle of $\Sdc$ in polar coordinates respectively. This figure shows that at an elevation angle denoted as $\tCoptdc$ the coverage radius reaches its maximum $\rCdc$. In order to find $\tCoptdc$ we simplify $\Lambda(\tC)$ in \eqref{Lambda} by proposing an approximate analytical solution for the inverse Marcum Q--function in the following lemma.

\begin{figure}[t!]
  \centering
  \begin{tikzpicture}
  \node[above right] (img) at (0,0) {\includegraphics[width=0.5\columnwidth]{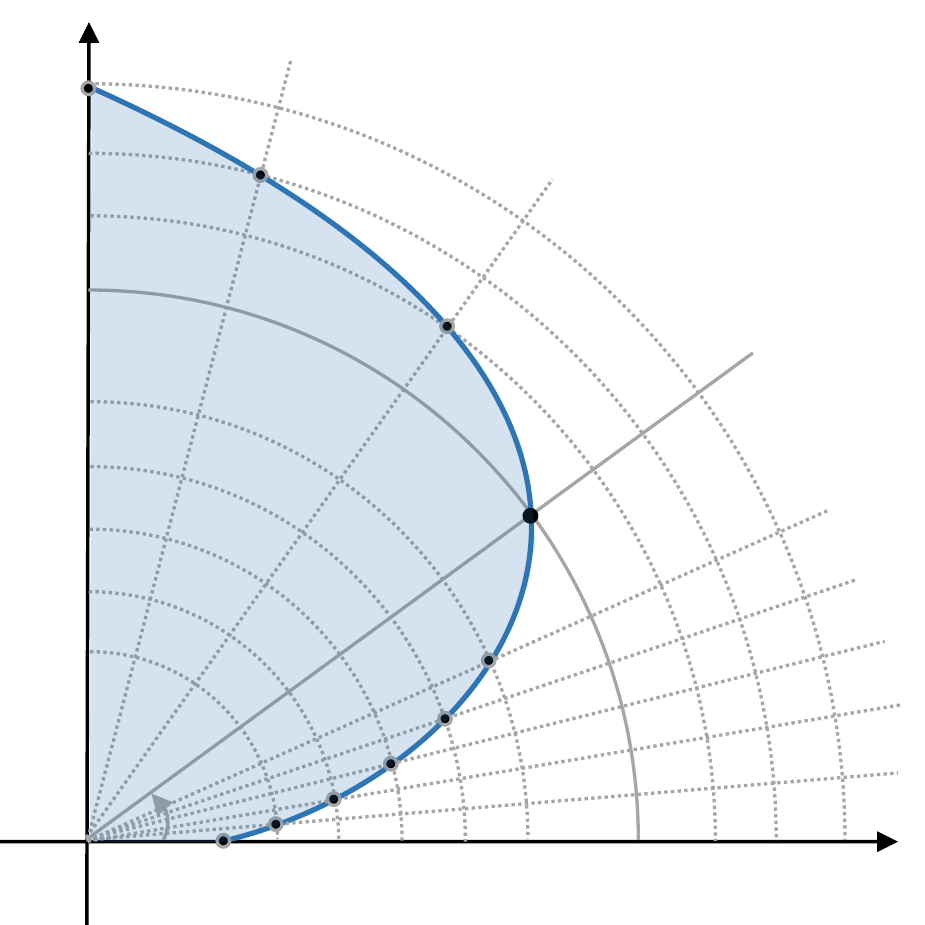}};
  \node at (260pt,25pt) {$\rD$};
  \node at (25pt,255pt) {$h$};
  \node at (60pt,35pt) {$\tCoptdc$};
  \node at (105pt,200pt) {$\Sdc$};
  \node at (80pt,150pt) {$\Poutdc < \varepsilon$};
  \node at (170pt,160pt) {$\Poutdc > \varepsilon$};
  \node at (175pt,110pt) {$(\rCdc,\hOPTdc)$};
  \node at (180pt,15pt) {$\Lambda(\tCoptdc)$};
  \node at (15pt,15pt) {O};
  \end{tikzpicture}
  \caption{Configuration space $\Sdc$ on $\rD$--$h$ plane.}\label{Sdc}
\end{figure}

\begin{lemma} \label{InvQ}
The inverse Marcum Q--function with respect to its second argument, i.e. $y = Q^{-1}(x,1-\varepsilon)$, is approximately given by
\begin{equation} \label{InvMQ}
y = 
     \begin{cases}
       \sqrt{-2\ln(1-\varepsilon)}~e^{\frac{x^2}{4}} & ;~x \leq x_0 \\
       x+\frac{1}{2Q^{-1}(\varepsilon)}\ln\left[\frac{x}{x-Q^{-1}(\varepsilon)}\right]-Q^{-1}(\varepsilon) & ;~x > x_0~\wedge~Q^{-1}(\varepsilon) \neq 0  \\
       x+\frac{1}{2x} & ;~x > x_0~\wedge~Q^{-1}(\varepsilon) = 0 \ 
     \end{cases}
\end{equation}
where $x_0$ is the intersection of the sub-functions at $x > \max[0,Q^{-1}(\varepsilon)]$ and $Q^{-1}(\cdot)$ is the inverse Q--function.
\end{lemma}

\begin{proof}
The proof is given in Appendix \ref{InvQ-proof}.
\end{proof}

\begin{corollary} \label{yFinalApp}
For $x \gg 1$, $y = Q^{-1}(x,1-\varepsilon)$ is approximately obtained as
\begin{equation}
y \cong x-Q^{-1}(\varepsilon).
\end{equation}
\end{corollary}

\begin{proof}
For $x \gg 1$, which results in $x > x_0$, one sees that $x \gg \frac{1}{2Q^{-1}(\varepsilon)}\ln\left[\frac{x}{x-Q^{-1}(\varepsilon)}\right]$ and $x \gg \frac{1}{2x}$, and hence the desired result is obtained. 
\end{proof}

Now using Corollary \ref{yFinalApp} the optimum elevation angle $\tCoptdc$ that yields the optimum altitude $\hOPTdc$ and the maximum coverage radius $\rCdc$ in \eqref{h,rC} is obtained through the following theorem. 

\begin{theorem} \label{OptEleAng theorem}
The approximate optimum elevation angle of the UAV $\tCoptdc$ at which the UAV has the maximum coverage radius $\rCdc$ can be find in the following implicit equation
\begin{equation} \label{OptEleAng}
\a(\tC)\tan(\tC)+\a'(\tC)\ln[\Lambda(\tC)] = 2\xc'\left(\frac{Q^{-1}(\varepsilon)}{\xc[\xc-Q^{-1}(\varepsilon)]}\right),
\end{equation}
where $\Lambda(\tC)$ and $\xc$ are defined in \eqref{h,rC} and $\a'(\tC)$ and $\xc'$ indicate the derivative functions with respect to $\tC$.
\end{theorem}

\begin{proof}
The proof can be found in Appendix \ref{OptEleAng theorem proof}.
\end{proof}

Notice that, by using the optimum elevation angle obtained from \eqref{OptEleAng} into the system equations in \eqref{h} and \eqref{rC}, the optimum altitude of the UAV $\hOPTdc$ and its maximum coverage radius $\rCdc$ can be obtained as
\begin{subequations}
\begin{align}
\hOPTdc &= \Lambda(\tCoptdc) \cdot \sin(\tCoptdc), \\
\rCdc &= \Lambda(\tCoptdc) \cdot \cos(\tCoptdc).
\end{align}
\end{subequations}
Moreover, \eqref{OptEleAng} shows that $\tCoptdc$ is independent from the transmit power $\PU$ and the SNR threshold $\xi$ provided that $\a'(\tCopt) \cong 0$ (c.f. Section \ref{numerical}). In other words, $\tCoptdc$ and $\hOPTdc/\rCdc = \tan(\tCoptdc)$ are only determined by $\varepsilon$ and the propagation parameters $K(\cdot)$ and $\a(\cdot)$ that are characterized by the type of environment and the system parameters such as carrier frequency. Therefore, using \eqref{h,rC} one finds that
\begin{equation} \label{prop}
\hOPTdc \propto \left(\frac{\PU}{\xi}\right)^{\frac{1}{\a(\tCoptdc)}},~~~\rCdc \propto \left(\frac{\PU}{\xi}\right)^{\frac{1}{\a(\tCoptdc)}}.
\end{equation} 


\section{Air-to-Ground Communication Using Ground Relaying} \label{relaying communication}

In this section a cooperative strategy with ground relaying is presented. The outage probability of the system is studied and an analytical lower bound is derived.


\subsection{Ground Decode-and-Forward Relaying}

We adopt the decode-and-forward (DF) opportunistic relaying method, where the data is transmitted to the destination D in two phases. In the first phase, the UAV broadcasts its data and provided that the SNR of the link between the UAV and an arbitrary relay node $\text{R}_j$ is high enough, the relay is able to successfully decode the received signal. These relay nodes form a set called $\A$ which may differ in each transmission attempt since the wireless channel between the UAV and ground nodes vary. In the second phase the best relay node $\text{R}_J$ in $\A$ which has the highest instantaneous SNR to the destination D is chosen to retransmit the received data to the destination. Therefore, we can write
\begin{equation}\label{relay startegy}
\A = \{\text{R}_j~|~\Gamma_{\text{U}{\text{R}_j}}>\xi\}~,~~~J \triangleq \arg\,\max\limits_{\text{R}_j \in \A}~{\Gamma _{\text{R}_j\text{D}}}~,
\end{equation}
where $\Gamma_{\text{U}{\text{R}_j}}$ and $\Gamma _{\text{R}_j\text{D}}$ are obtained from \eqref{GUX} and \eqref{G_G1G2} respectively. The outage probability of the relaying communication (rc) can be defined as
\begin{equation}\label{out-def-R}
\Poutrc = \mathds{P}\left(\Gamma _{\text{R}_J\text{D}}\leq\xi\right),
\end{equation}
which is obtained in the following theorem.

\begin{theorem} \label{theorem - relaying outage}
The outage probability of the relaying communication $\Poutrc$ can be written as
\begin{subequations} \label{Relaying Outage}
\begin{equation}
\Poutrc(\rD,h) = e^{-\lambda[\psi_1(h)-\psi_2(\rD,h)]}
\end{equation}
where
\begin{equation}
\psi_1(h) = 2\pi \int_0^{\rrc} \rR~Q\left(\sqrt{2K(\tR)},\sqrt{{2[K(\tR)+1]\,\xi {\lUR^{\,\a(\tR)}}}/{\gamma_\text{U}}}\right)\,d\rR,
\end{equation}
\begin{align}
\psi_2(\rD,h) = \int_0^{2\pi}\!\!\!\int_0^{\rrc} &\rR\left[1-Q\left(\sqrt{2\kappa_0},\sqrt{{2(\kappa_0+1)\,\xi
\lRD^{\,\a_0}}/{\gamma_\text{R}}}\right)\right] \\ 
&\times  Q\left(\sqrt{2K(\tR)},\sqrt{{2[K(\tR)+1]\,\xi
\lUR^{\,\a(\tR)}}/{\gamma_\text{U}}}\right) d\rR \, d\varphi_\text{R},
\end{align}
\begin{equation}
\lUR = \sqrt{r_\text{R}^2+h^2},~~ \lRD = \sqrt{r_\text{R}^2+r_\text{D}^2-2r^{}_\text{R}r^{}_\text{D}\cos(\varphi_\text{R}-\varphi_\text{D})} ,
\end{equation}
\begin{equation}
\tR = \tan^{-1}(h/\rR),~~\gamma_\text{R} = \frac{A P_\text{R}}{N_0}.
\end{equation} 
and $\rrc$ indicates the radius of the coverage area $\mathcal{C}$ for the relaying communication.
\end{subequations} 
\end{theorem} 

\begin{proof}
The proof is analogous to \cite{wang2011general} and is given in Appendix \ref{theorem - relaying outage proof}.
\end{proof}

The relation in \eqref{Relaying Outage} shows that the relaying outage probability $\Poutrc(\rD,h)$ exponentially decreases with the density of the relays $\lambda$. In the following corollary a lower bound for $\Poutrc(\rD,h)$ simplifies the corresponding expression.

\begin{corollary}
The outage probability of the relaying communication in \eqref{Relaying Outage} is lower bounded as
\begin{equation}
\Poutrc(\rD,h) \geqslant e^{-\lambda[\psi_{01}-\psi_{02}(\rD)]},
\end{equation}
where $\psi_{01} = |\C|$ is the area of $\C$, and
\begin{equation}
\psi_{02}(\rD) = \int_0^{2\pi}\!\!\!\int_0^{\rrc} \rR \left[1-Q\left(\sqrt{2\kappa_0},\sqrt{{2(\kappa_0+1)~\xi
\lRD^{\,\a_0}}/{\gamma_\text{R}}}\right)\right] d\rR\,d\fiR.
\end{equation}
\end{corollary}

\begin{proof}
Assume that in the first phase all the relay nodes over $\C$ are able to decode the received signal from the UAV. In this case by denoting the outage probability as $\mathcal{P}_\text{out}^{\text{LB}}$ we have $\Poutrc \geqslant \mathcal{P}_\text{out}^{\text{LB}}$. To compute $\mathcal{P}_\text{out}^{\text{LB}}$ we notice that the aforementioned assumption leads to $\lambda_{\A} = \lambda$ and hence
\begin{equation} \label{mu-sim}
\mu_{\A} = \int_{\C} \lambda_{\A}\,d\C = \lambda |\C|.
\end{equation} 
In addition \eqref{Gr1d<xi} can be simplified to
\begin{subequations}\label{Gr1d<xi-sim}
\begin{align}
\mathds{P}\left(\Gamma _{\text{R}\text{D}}\leq\xi\right) &= \int_{\C} \mathds{P}\left(\Gamma _{\text{R}\text{D}}\leq\xi~|~\text{R}: (r_\text{R},\varphi_\text{R})\right) \frac{\lambda}{\mu_{\A}} ~d\C \\ \label{Gr1d<xi-sim-1}
&= \frac{1}{|\C|}\int_{\C} \mathds{P}\left(\frac{A P_{\text{R}}}{N_0{\lRD}^{\!\!\a_0}}~\Omega_{{\text{R}}{\text{D}}}\leq\xi~|~\text{R}: (r_\text{R},\varphi_\text{R})\right)~d\C \\ \label{Gr1d<xi-sim-2}
&= \frac{1}{|\C|} \int_0^{2\pi}\!\!\!\int_0^{\rrc} \rR \left[1-Q\left(\sqrt{2\kappa_0},\sqrt{{2(\kappa_0+1)~\xi
\lRD^{\,\a_0}}/{\gamma_\text{R}}}\right)\right]~d\rR\,d\fiR.
\end{align}
\end{subequations}
Therefore, by using \eqref{mu-sim} and \eqref{Pout-final-2} one obtains
\begin{equation}
\mathcal{P}_\text{out}^{\text{LB}} = e^{-\lambda[\psi_{01}-\psi_{02}(\rD)]}.
\end{equation} 
\end{proof}

The proposed lower bound for $\Poutrc$ can be reached at the altitudes of the UAV where it has a good channel condition only with the relay nodes in vicinity of the destination. This is due to the fact that the best relay in the second phase is more likely to be chosen among the candidates located near the destination enduring lower path loss. Therefore, although the mathematical solution for $\mathcal{P}_\text{out}^{\text{LB}}$ is based on the assumption that all the relay nodes successfully decode the UAV's signal in the first phase, considering the second phase of communication and the above-mentioned opportunistic relaying strategy, if only the relay nodes in the proximity of the destination successfully decode the transmitted signal the lower bound provides a tight approximation of the actual outage probability. This fact suggests a range of altitudes at which the UAV can reach the lowest outage probability in relaying strategy. This is further explored in Section \ref{numerical}.   


\subsection{Cooperative Communication}

In an opportunistic relaying cooperative network, the destination D receives the transmitted signal by the UAV from both the direct and the best relay path in the first and second phase respectively. Considering the selection combining strategy in which only the received signal with the highest SNR at the destination is selected, the total outage probability of the cooperative communication $\Poutcc$ can be written as
\begin{align}\label{Tot Outage}
\Poutcc = \mathds{P}\left(\max\{\GUD,\Gamma _{\text{R}_J\text{D}}\} \leq \xi \right) = \Poutdc \cdot \Poutrc,
\end{align}
where the last equation is due to the independency of fading between any pair of nodes. By substituting $\Poutdc$ and $\Poutrc$ from \eqref{Pout=1-Q} and \eqref{Relaying Outage} into \eqref{Tot Outage} the total outage probability is obtained, which is a function of $\rD$ and $h$, i.e. $\Poutcc = \Poutcc(\rD,h)$. Note that in \eqref{Relaying Outage} $\rCrc$ is to be replaced with $\rCcc$ which is the radius of $\C$ for cooperative communication.

%
%

Due to the complicated expressions of $\Poutrc$ and $\Poutcc$ in \eqref{Relaying Outage} and \eqref{Tot Outage}, the problems \eqref{h*Y} and \eqref{rCY} for relaying and cooperative communication are not mathematically tractable. However, numerical optimization for particular scenarios is possible. An example of this is presented in Section \ref{numerical-B}. 


\section{Case Study: \\ Special Dependency of $\a$ and $K$ Over $\theta$} \label{numerical}

In this section we propose specific relations for $\a(\theta)$ and $K(\theta)$ needed for numerical results.


\subsection{Models for $\a(\theta)$ and $K(\theta)$}

The value of path loss exponent is typically proportional to the density of obstacles between the transmitter and receiver such that a larger $\a$ is assumed in denser areas. Therefore, $\a(\theta)$ can be characterized using the notion of probability of line of sight (LoS) $\mathcal{P}_\text{LoS}(\theta)$ between the UAV and the ground node \cite{al2014optimal}. This relationship is defined as
\begin{equation}\label{alfa-Plos}
\a(\theta) = a_1 \cdot \mathcal{P}_\text{LoS}(\theta) + b_1,
\end{equation}
where 
\begin{equation}\label{PrLoS}
\PLoS(\theta) = \frac{1}{1+a_2 e^{-b_2 \theta}},
\end{equation}
and $a_1$, $b_1$, $a_2$ and $b_2$ are determined by the environment characteristics and the transmission frequency, and $\theta$ is in radian. A direct calculation shows that
\begin{align}
a_1 &= \frac{\a_{\frac{\pi}{2}}-\a_0}{\mathcal{P}_\text{LoS}(\frac{\pi}{2})-\mathcal{P}_\text{LoS}(0)} \cong \a_{\frac{\pi}{2}}-\a_0, \nonumber \\
b_1 &= \a_0 - a_1 \cdot \mathcal{P}_\text{LoS}(0) \cong \a_0,
\end{align}
where the approximations are due to the fact that $\mathcal{P}_\text{LoS}(0) \rightarrow 0$ and $\mathcal{P}_\text{LoS}(\frac{\pi}{2}) \rightarrow 1$. The proposed model for $\a(\theta)$ in \eqref{alfa-Plos} is further discussed relying on the recent reports in Appendix \ref{PathLossExp}.

For the Rician factor, in consistency with \cite{shimamoto2006channel}, we follow the exponential dependency between $K$ and $\theta$ as
\begin{equation}
K(\theta) = a_3 \cdot e^{b_3 \theta},
\end{equation}
where $\theta$ is in radian and $a_3$ and $b_3$ are environment and frequency dependent constant parameters which are related to $\kappa_0$ and $\kappa_{\pi/2}$ as
\begin{align}
a_3 &= \kappa_0, \nonumber \\
b_3 &= \frac{2}{\pi} \ln\left(\frac{\kappa_{\frac{\pi}{2}}}{\kappa_0}\right).
\end{align}
Note that the general shape of $\a(\theta)$ and $K(\theta)$ can change over different environments and system parameters that could be specified by the measurements in concrete scenarios. 


\subsection{Simulation and Discussion} \label{numerical-B}

In this section the simulations are provided to discuss the analytical results obtained in the previous sections. 

\subsubsection{With the Same UAV Transmit Power}

The parameters used in this subsection are set to $\gamma_\text{U} = \gamma_\text{R} = 75$ dB, $\kappa_0 = 5$ dB, $\kappa_{\frac{\pi}{2}} = 15$ dB, $\a_0 = 3.5$, $\a_{\frac{\pi}{2}} = 2$, and $\lambda = 0.0003$, unless otherwise indicated. 

Figure \ref{Pout_h_rD} shows that the analytical results of the outage probability are in a good conformity with the simulation results for each of direct, relaying and cooperative communication in which $10^5$ independent network realizations are employed for the simulations. As can be seen, comparing with G2G communication, i.e. $h = 0$, the outage performance is significantly enhanced by exploiting UAV at appropriate altitudes. Moreover, the outage probabilities as function of altitude are convex resulting in the existence of altitudes $\hat{h}^\text{Y} = \hdc, \hrc, \hcc$, or the corresponding elevation angles $\hat{\theta}^{\,\text{Y}}_\text{D} = \tetdc, \tetrc, \tetcc$, where the outage probabilities are minimized. In fact for $h<\hat{h}^\text{Y}$ the benefits of the reduced path loss exponent $\a$ and the increased Rician factor $K$ with increasing $h$ becomes more significant than the losses caused by the increased link length and hence the outage probability decreases. However beyond $\hat{h}^\text{Y}$ the impact of the link length dominates the other factors and hence leads to an increased outage probability. Therefore, at $h = \hat{h}^\text{Y}$ the impact of the above-mentioned factors are balanced resulting in the minimum outage probability. This altitude increases with the distance $\rD$ as can be seen in the figure. 

\begin{figure}[t!]
  \centering
  \includegraphics[width=0.6\columnwidth]{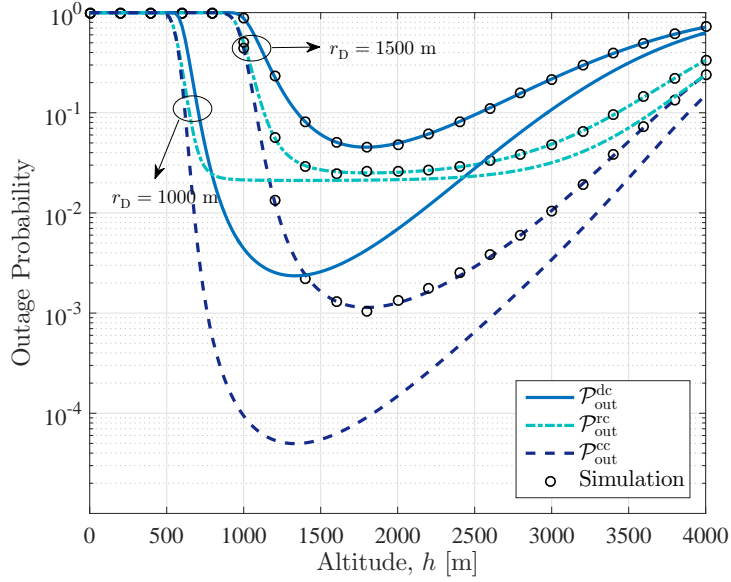}\\
  \caption{Outage probability is a convex function of altitude where optimal altitude is mainly determined by $\Poutdc$. Because of the low transmit power considered in this simulation, the outage 
  probability at low altitudes is very high, yet when optimizing the altitude we can achieve a very good outage performance with the same transmit power.} \label{Pout_h_rD}
\end{figure}

\begin{figure}[t!]
  \centering
  \includegraphics[width=0.6\columnwidth]{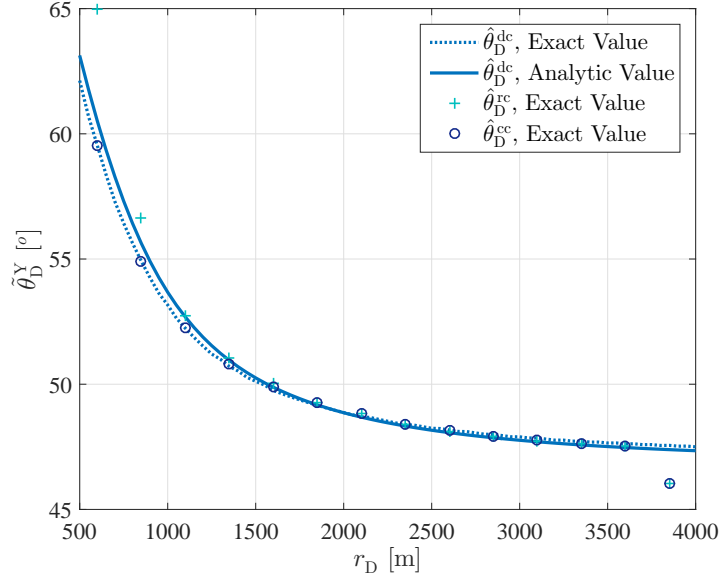}\\
  \caption{The analytical solution for $\tetdc$ closely matches the value obtained by the numerical simulation. Moreover, $\tetcc$ is mainly determined by the direct communication link such that $\tetcc = \tetdc$ except near the border of $\C$.}\label{tet_hat_all_rD}
\end{figure}

Is interesting to note that, according to Figure \ref{Pout_h_rD}, the relaying communication outage probability $\Poutrc$ might be lower or higher than that of the direct communication $\Poutdc$ depending on the altitude $h$ and the destination distance $\rD$. For instance at $\rD = 1000$ m, $\Poutdc$ is bigger than $\Poutrc$ at very low and very high altitudes, however for moderate altitudes the direct communication performs better than the relaying communication. Indeed, the relaying outage probability is limited by the second phase of the corresponding communication strategy where G2G link endures a larger path loss exponent. Due to this fact, $\Poutrc$ remains approximately constant over a range of UAV's altitude at which the A2G link has a good quality and hence the overall relaying outage performance is determined by the G2G link. Therefore, the outage probability of cooperative communication $\Poutcc$, which is the product of $\Poutdc$ and $\Poutrc$, is minimized approximately at the altitude of the UAV where $\Poutdc$ reaches its minimum. In other words $\hdc \cong \hcc$ and hence $\tetdc \cong \tetcc$. This fact is illustrated in Figure \ref{tet_hat_all_rD}.

\begin{figure}[t!]
  \centering
  \includegraphics[width=0.6\columnwidth]{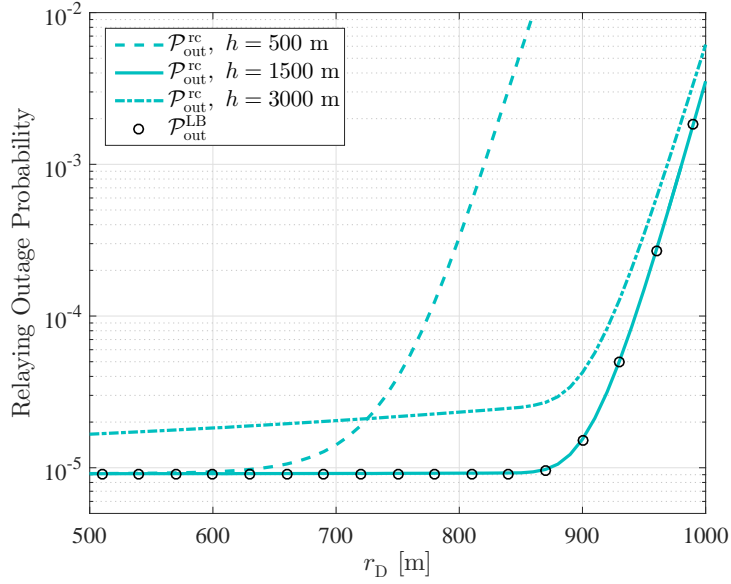}\\
  \caption{The proposed lower bound outage probability $\Pout^\text{LB}$ is very tight over the entire range of $r_\text{D}$ when the UAV is in appropriate altitudes.}\label{Pout_Exact_LB_rD_h}
\end{figure}

The values of $\tetdc$ obtained from Theorem \ref{t-hat} are compared with the exact (numerical) values in Figure \ref{tet_hat_all_rD} which shows the accuracy of the analytical solution proposed in the theorem. As can be seen, $\tetdc$ reduces with $\rD$ since the link length is more susceptible to the elevation angle at larger $\rD$ and hence the beneficial effect of reduction in $\a$ and increase in $K$ becomes less noticeable compared to the additional link length. However $\tetdc$ is approximately independent from $\rD$ at larger distances. According to Figure \ref{tet_hat_all_rD}, although $\tetrc$ may be different with $\tetdc$ at low altitudes, $\tetcc$ is equal to $\tetdc$ which means that the equation \eqref{tD*Eq} is valid for cooperative communication as well. However close to the border of $\C$, both $\tetrc$ and $\tetcc$ deviate from $\tetdc$ owing to the fact that the candidate relays for cooperation are limited to the region between the UAV and the destination D and hence the elevation angle for minimum outage probability decreases.

\begin{figure}[t!]
  \centering
  \includegraphics[width=0.6\columnwidth]{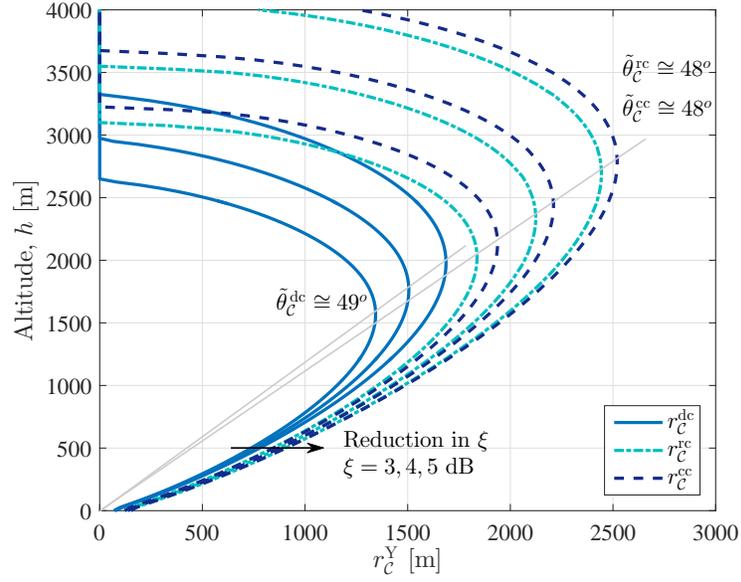}\\
  \caption{Configuration space $\mathcal{S}^\text{Y}$ for different type of communications. The optimum elevation angle $\tilde{\theta}_\mathcal{C}^{\,\text{Y}}$ is independent from the SNR requirement $\xi$.}\label{CovRadius_all_h_xi}
\end{figure}

Figure \ref{Pout_Exact_LB_rD_h} shows that the proposed lower bound for relaying communication is tight over the entire range of $r_\text{D}$ at the optimal altitudes. This is due to the fact that the relay nodes around the neighboring of destination are well connected to the UAV and thus the outage event only occurs in the second phase of relaying communication. The tightness of lower bound outage probability enables us to use it instead of the complex exact expression for $\Poutrc$ obtained in the previous section.

The locus of the configuration space $\mathcal{S}^\text{Y}$ for $\text{Y} = \text{dc}$, $\text{rc}$, $\text{cc}$ are depicted in Figure \ref{CovRadius_all_h_xi}. As can be seen, the impact of third dimension, i.e. altitude, by exploiting UAV is striking in order to extend the coverage range. Furthermore, there exists an optimum altitude $\tilde{h}_\mathcal{C}^\text{Y}$ and elevation angle $\tilde{\theta}_\mathcal{C}^{\,\text{Y}}$ which leads to the maximum coverage radius $\tilde{r}_\mathcal{C}^{\,\text{Y}}$. The figure shows that $\tilde{\theta}_\mathcal{C}^{\,\text{Y}}$ is independent from the SNR requirement $\xi$ and hence the ratio of ${\tilde{h}_\mathcal{C}^\text{Y}}/{\tilde{r}_\mathcal{C}^{\,\text{Y}}}$ is constant. However, $\tilde{h}_\mathcal{C}^\text{Y}$ and $\tilde{\theta}_\mathcal{C}^{\,\text{Y}}$ diminish with $\xi$ as expressed in \eqref{prop}. Figure \ref{CovRadius_all_h_xi} also shows that the relaying communication might finally result in a coverage area larger than that of direct communication, although $\Poutrc$ is higher than $\Poutdc$ at some altitudes $h$ and distances $\rD$ which can be seen in Figure \ref{Pout_h_rD}.

\subsubsection{With the Same Total Transmit Power Budget}

In this subsection we compare the results while adopting the same total transmit power budget at the transmitting nodes, i.e. the UAV and the relay. In other words we assume a total power budget of $P_\text{T}$ and hence we assign $\PU = P_\text{T}$ for direct communication and $\PU+\PR = P_\text{T}$ for relaying and cooperative communication so that the comparisons are fair. However, the fundamental concern here is how the power budget is to be assigned to each of the transmitting nodes in order to increase the coverage range. To this end, we define the power allocation factor $\rho$ as
\begin{equation}
\rho = \frac{\PU}{\PU+\PR},
\end{equation}
representing the portion of power budget allocated to the UAV. Therefore, in the following we discuss the maximum coverage obtained by using the optimum $\rho$ in each altitude. The optimum power allocation factor for relaying and cooperative strategy are denoted as $\rOrc$ and $\rOcc$ respectively. 

\begin{figure}
\centering
\begin{subfigure}{.5\textwidth}
  \centering
  \includegraphics[width=\linewidth]{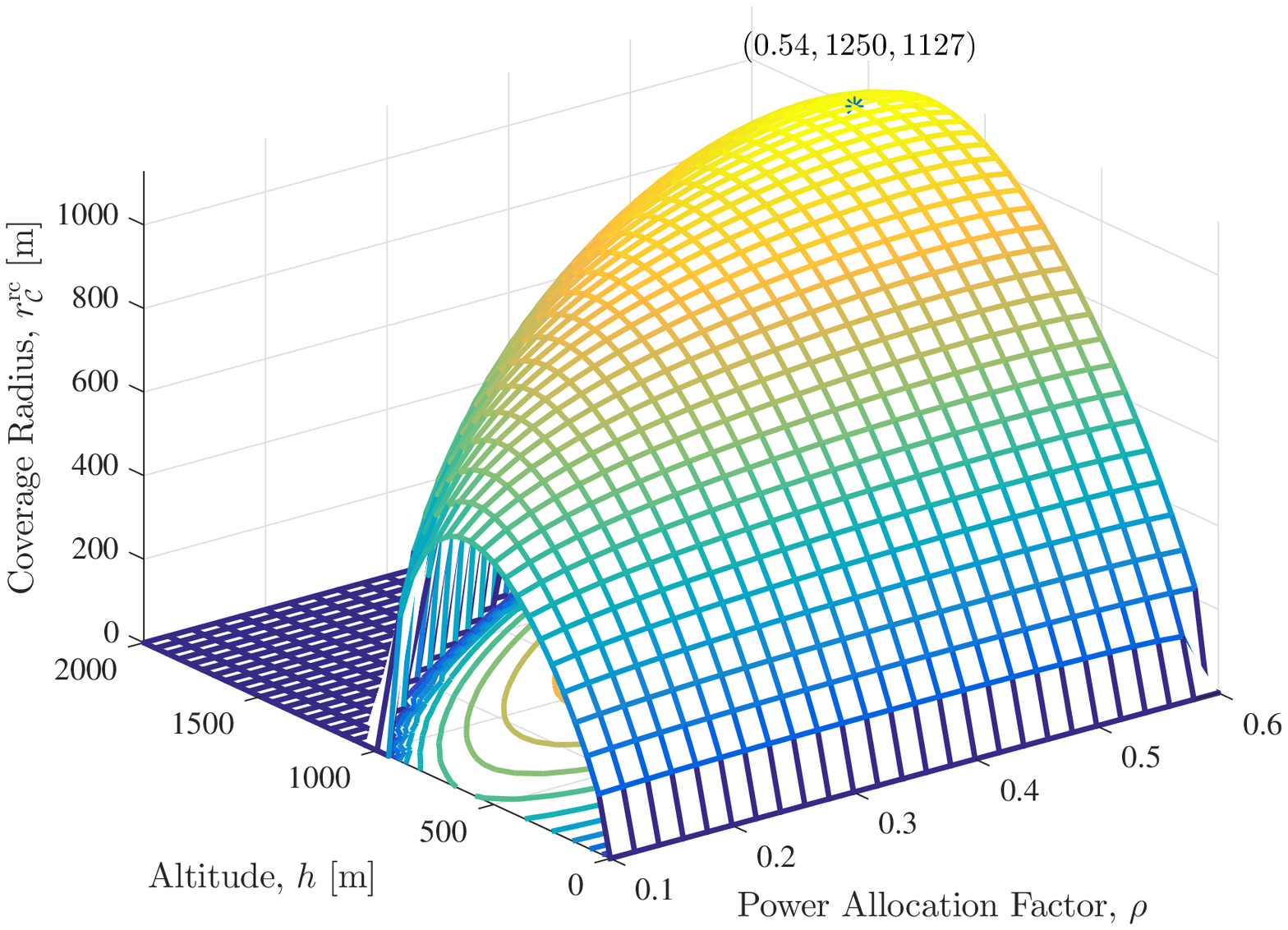}
  \caption{}
  \label{CovRadius_rc_betta_h_3D}
\end{subfigure}%
\begin{subfigure}{.5\textwidth}
  \centering
  \includegraphics[width=\linewidth]{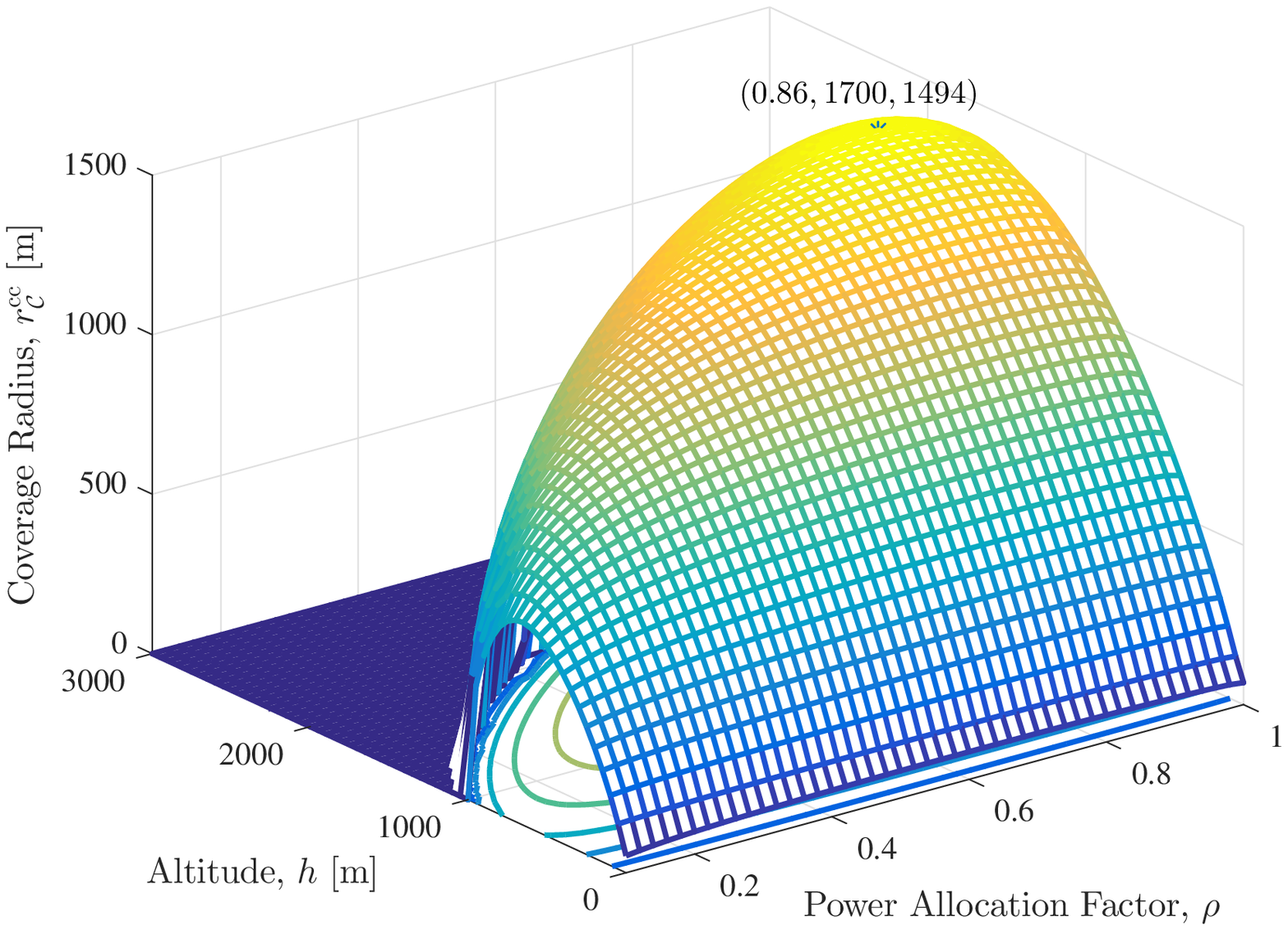}
  \caption{}
  \label{CovRadius_cc_betta_h_3D}
\end{subfigure}
\caption{(a) The coverage radius is maximized at the optimum alitude by allocating the optimum portion of transmission power budget to the UAV. (b) The optimum power allocation factor in cooperative communication is larger than that of the relaying strategy.}
\label{CovRadius_y_betta_h_3D}
\end{figure}

%

Figures \ref{CovRadius_y_betta_h_3D} shows that the coverage radius is a concave function of power allocation factor $\rho$ at each altitude which results in an optimum $\rho$ maximizing the coverage. This means that the performance of the network is maximized when the available power budget is optimally assigned to each of two phases. The same behavior can be observed while looking at coverage radius as function of altitude for a given $\rho$. This fact leads to a unique optimum $h$ and $\rho$ for the maximum coverage radius as is marked in the figures. Generally speaking, at each altitude, allocating more power to the UAV is more effective than to the relays since the communication channel between the UAV and a ground terminal suffers from less path loss than a channel between a relay and a ground destination D, and hence $\rOrc > 0.5$. The optimum power allocation factor for cooperative communication, i.e. $\rOcc$, is even larger than $\rOrc$ since the UAV's signal is also received at the destination which increases the contribution of UAV's transmit power to the final received SNR and hence to the coverage radius.

\begin{figure}
\centering
\begin{subfigure}{.5\textwidth}
  \centering
  \includegraphics[width=\linewidth]{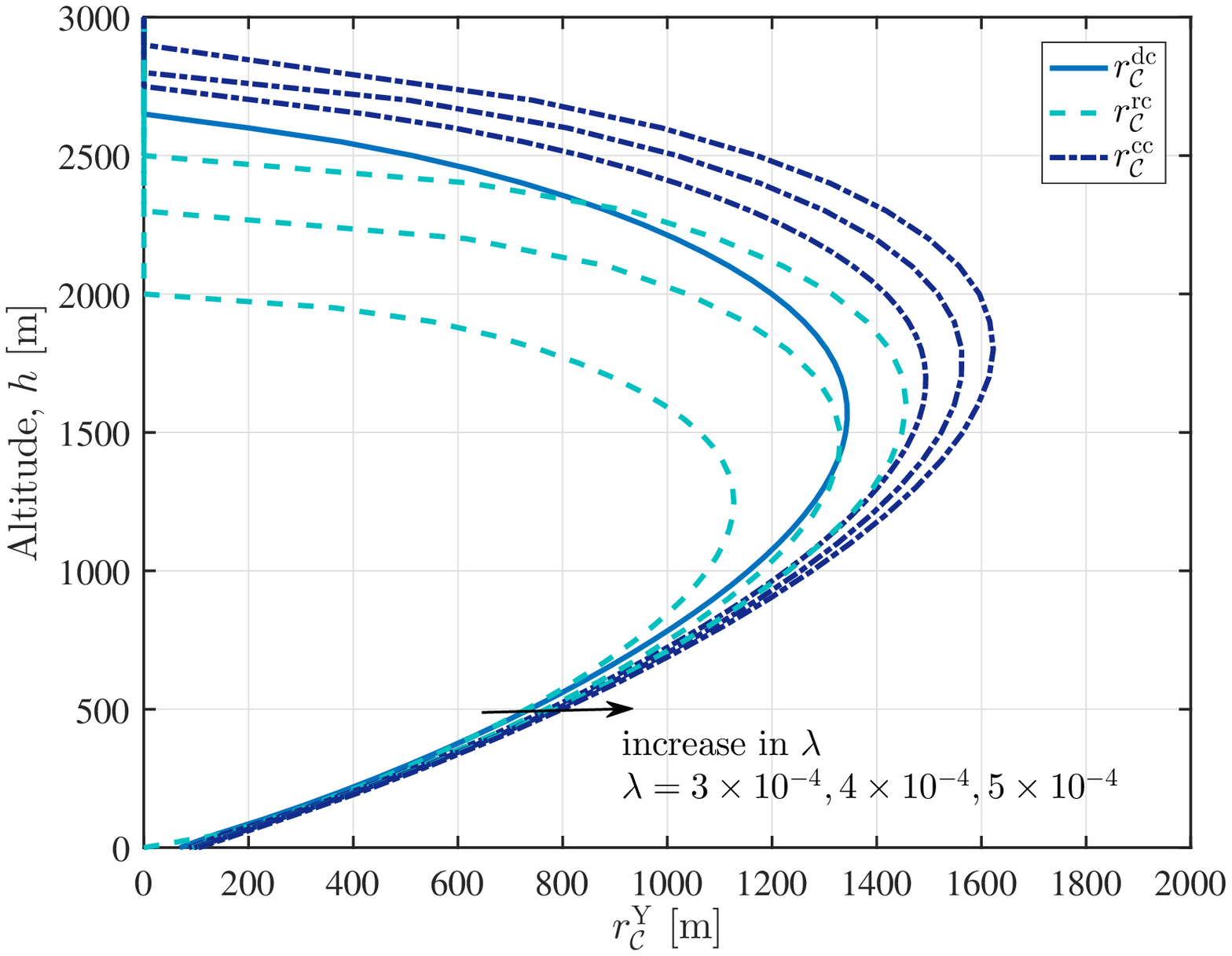}
  \caption{}
  \label{CovRadius_all_h_lambda_powAllo}
\end{subfigure}%
\begin{subfigure}{.5\textwidth}
  \centering
  \includegraphics[width=\linewidth]{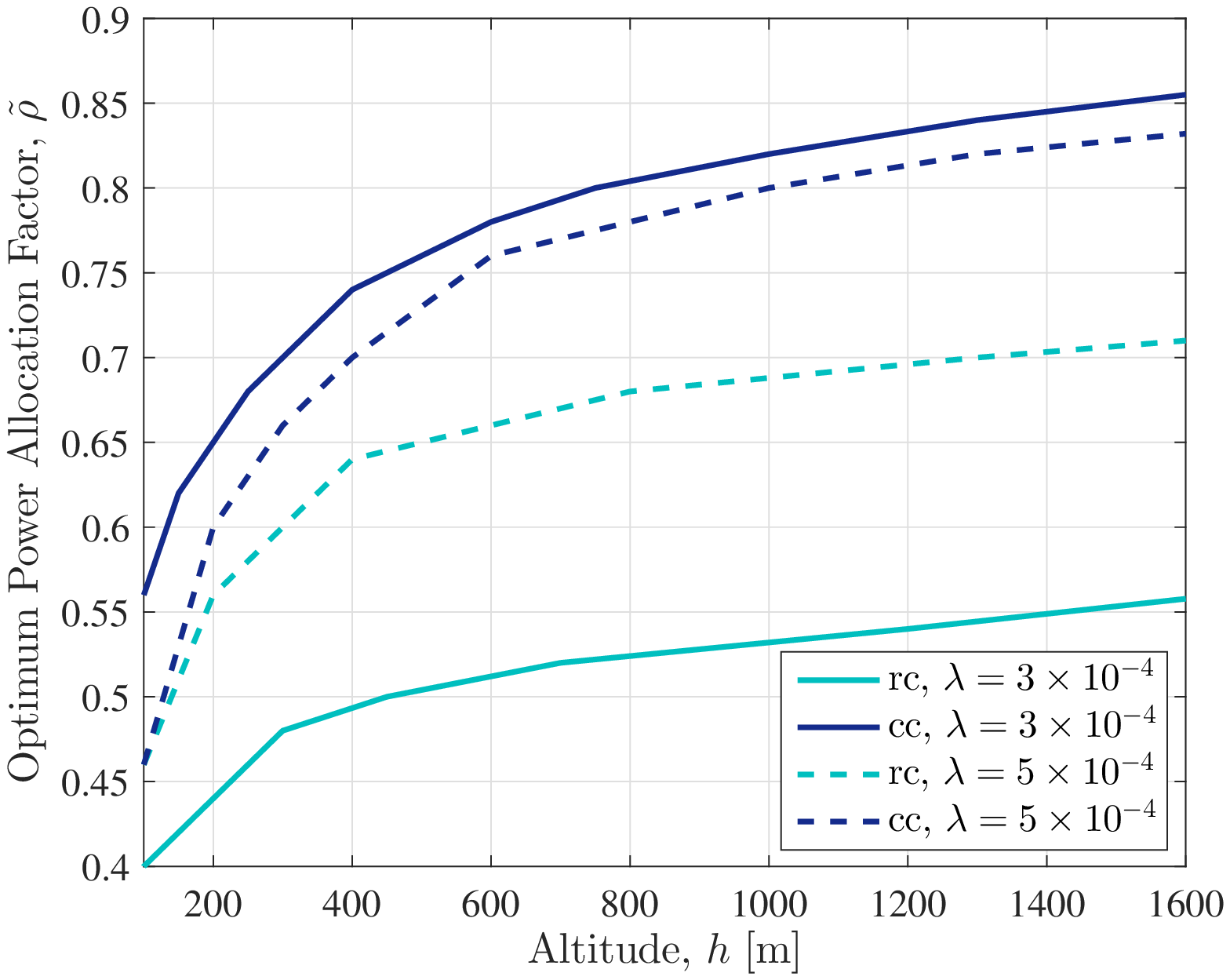}
  \caption{}
  \label{OptPowAllFact_all_h_lambda}
\end{subfigure}
\caption{(a) Even using the same transmission power budget, cooperative communication reults in a larger coverage area. (b) Optimum ratio of power allocation $\tilde{\rho}_\mathcal{C}$ increases with altitude $h$. It is interesting to see that for lower altitude (which is the realistic region given regulatory constraints that limit UAVs to fly too high), the optimal power allocation strategy gives less power to the UAV. This is good, as this means that the power requirement for those low altitude UAVs is lower, however, means also that a ground relay communication strategy is important.}
\label{all_h_lambda}
\end{figure}

%
%

Figure \ref{CovRadius_all_h_lambda_powAllo} shows the coverage radius of the UAV at each altitude while employing the optimum $\rho$. As can be seen, depending on the density of relays, $r_\mathcal{C}^\mathrm{rc}$ could be larger than $r_\mathcal{C}^\mathrm{dc}$ even though the UAV transmit power is lower. Therefore, the relaying strategy might perform better than the direct communication while adopting the same overall transmit power budget. Moreover, cooperative communication always results in a larger coverage area independent from the density of relays $\lambda$. This is due to the fact that the portion of relaying transmit power is determined based on the availability of the relay nodes such that for denser areas more power is dedicated to the relays and hence optimal $\rho$ decreases as $\lambda$ increases (see Figure \ref{OptPowAllFact_all_h_lambda}). However, for relaying communication optimal $\rho$ grows as $\lambda$ becomes larger. As is shown in Figure \ref{OptPowAllFact_all_h_lambda}, the optimum power allocation factors, i.e. $\rOrc$ and $\rOcc$, are an increasing function of altitude $h$. Therefore, as the UAV goes higher more portion of power budget is to be assigned to the UAV for maximum coverage range.



\section{Conclusion} \label{conclusion}

The analysis of UAV networks requires a detailed model of the A2G communication links as function of distance and height. A generic A2G analytical framework was proposed, which takes into account the dependence of the path loss exponent and multipath fading on the height and angle of the UAV. We showed how this model enables a characterization of the performance and reliability of A2G cooperative communication networks.

Moreover, the model enables to derive the UAV altitude that maximizes the reliability and coverage range, which is crucial for UAV usage scenarios such as aerial sensing or data streaming. Results show that the optimal altitude for maximum reliability is mainly determined by the direct link of communication in an A2G cooperative system, as the relaying communication achieves a stable performance for a wide range of UAV altitudes. For a specific scenario with a communication range of 1000$m$ with a low transmit power, the optimal UAV altitude was shown to be 1300$m$ for the direct communication, while a range of 700$m$ to 2000$m$ approximated the optimal UAV altitude for the relaying communication.

Furthermore, by constraining the total transmit power budget, our results give insight in the dimensioning of the system. For example, for a UAV altitude of 200$m$ that complies with regulatory constraints, our analysis shows that an allocation of 35\% lower transmit power to the UAV, it is possible to obtain the same coverage range, thanks to the inclusion of ground relays. Alternatively, if in a similar scenario the UAV height is unconstrained, the coverage range increases up to 25\% compared to the direct communication at the optimum altitudes yet with even 15\% lower UAV transmit power. In fact, this reduction of the UAVs transmit power is of practical importance, first because of the limitation in the source of UAVs power. Second, the UAVs due to the higher LoS probability impose more interference to the ground users compared to the terrestrial interferers and hence their transmit power is to be lowered. This fact will be more investigated in our future study.



\appendices

\section{Proof of Theorem \ref{t-hat}} \label{t-hat proof}

For the elevation angle of the UAV $\tetdc$ at which the link's outage probability $\Poutdc(\rD,h)$ is minimized, one can write
\begin{equation} \label{Pout'=0}
\frac{\partial}{\partial \tD} \Poutdc(\rD,h) = 0.
\end{equation}
By defining the auxiliary variables $x$ and $y$ as
\begin{subequations}\label{x,y-tD}
\begin{align} \label{x-tD}
x &= \sqrt{2 K(\tD)}, \\ \label{y-tD}
y &= \sqrt{2\xi~[1+K(\tD)]~{\lUD}^{\a(\tD)}/\gamma_\text{U}},
\end{align}
\end{subequations}
which are the first and second arguments of the Marcum Q--function in \eqref{Pout=1-Q} respectively, and replacing $\Poutdc(\rD,h)$ from \eqref{Pout=1-Q} into \eqref{Pout'=0} we have
\begin{subequations}
\begin{align}
0 &= \frac{\partial}{\partial \tD} Q(x,y) \\ \label{Qderivation}
  &= \frac{\partial\,Q(x,y)}{\partial x} \cdot \frac{\partial x}{\partial \tD} + \frac{\partial\,Q(x,y)}{\partial y} \cdot \frac{\partial y}{\partial \tD}.
\end{align}
\end{subequations}
From \cite{short2012computation} one sees that
\begin{subequations}\label{Qprop1-both}
\begin{align}\label{Qprop1-1}
\frac{\partial\,Q(x,y)}{\partial x} &= y\,e^{-\frac{x^2+y^2}{2}}\,I_1(xy), \\ \label{Qprop1-2}
\frac{\partial\,Q(x,y)}{\partial y} &= -y\,e^{-\frac{x^2+y^2}{2}}\,I_0(xy).
\end{align}
\end{subequations}
Substituting \eqref{Qprop1-both} into \eqref{Qderivation} yields
\begin{equation}
y\,e^{-\frac{x^2+y^2}{2}}\,I_1(xy) \cdot \frac{\partial x}{\partial \tD} - y\,e^{-\frac{x^2+y^2}{2}}\,I_0(xy) \cdot \frac{\partial y}{\partial \tD} = 0,
\end{equation}
or equivalently
\begin{equation}\label{y'=x'I/I}
\frac{\partial y}{\partial \tD} = \frac{I_1(xy)}{I_0(xy)} \cdot \frac{\partial x}{\partial \tD}.
\end{equation}
Using \eqref{x-tD} one can write
\begin{equation}\label{x'}
\frac{\partial x}{\partial \tD} = \frac{K'(\tD)}{\sqrt{2 K(\tD)}} = \frac{K'(\tD)}{x},
\end{equation}
where $K'(\cdot)$ indicates the derivative function of $K(\cdot)$. In order to compute the derivative of $y$ with respect to $\tD$ in \eqref{y-tD} we notice that $\lUD = \rD/\cos(\tD)$. Thus we can write
\begin{equation}
\ln(y) = \frac{1}{2} \left[\ln\left(\frac{2\xi}{\gamma_U}\right)+\ln[1+K(\tD)]+\a(\tD)\ln\left(\frac{\rD}{\cos(\tD)}\right) \right].
\end{equation}
Taking the derivative with respect to $\tD$ in the above equation yields
\begin{equation}
\frac{\frac{\partial y}{\partial \tD}}{y} = \frac{1}{2} \left[\frac{K'(\tD)}{1+K(\tD)}+\a'(\tD)\ln\left(\frac{\rD}{\cos(\tD)}\right)+\a(\tD)\tan(\tD) \right],
\end{equation}
or equivalently
\begin{equation} \label{y'}
\frac{\partial y}{\partial \tD} = \frac{y}{2} \left[\frac{K'(\tD)}{1+K(\tD)}+\a'(\tD)\ln\left(\frac{\rD}{\cos(\tD)}\right)+\a(\tD)\tan(\tD) \right].
\end{equation}
By using \eqref{y'=x'I/I}, \eqref{x'} and \eqref{y'} one can write
\begin{equation}\label{xy/2}
\frac{xy}{2} \left[\frac{K'(\tD)}{1+K(\tD)}+\a'(\tD)\ln\left(\frac{\rD}{\cos(\tD)}\right)+\a(\tD)\tan(\tD) \right] = \frac{I_1(xy)}{I_0(xy)} K'(\tD).
\end{equation}
Now assuming that $xy$ at $\tD = \tetdc$ is large enough, we use the following approximation \cite{segura2011bounds}
\begin{align}\label{I/I}
\frac{I_1(xy)}{I_0(xy)} = 1-\frac{1}{2xy}-\frac{1}{8(xy)^2}+\mathcal{O}\left[(xy)^{-3}\right] \cong 1.
\end{align}
On the other hand, assuming that $K(\tetdc) \gg 1$, from \eqref{x,y-tD} one obtains
\begin{equation}\label{xy=}
xy \cong  2K(\tD) \sqrt{\frac{\xi}{\gamma_\text{U}}\left[\frac{\rD}{\cos(\tD)}\right]^{\a(\tD)}}.
\end{equation}
Therefore, from \eqref{xy/2}, \eqref{I/I} and \eqref{xy=} and using the assumption of $K(\tetdc) \gg 1$ we finally obtain
\begin{equation}
 \sqrt{\frac{\xi}{\gamma_\text{U}}\left[\frac{\rD}{\cos(\tD)}\right]^{\a(\tD)}} \left[\frac{K'(\tD)}{K(\tD)}+\a'(\tD)\ln\left(\frac{\rD}{\cos(\tD)}\right)+\a(\tD)\tan(\tD) \right] =  \frac{K'(\tD)}{K(\tD)}.
\end{equation}


\section{Proof of Lemma \ref{InvQ}} \label{InvQ-proof}

Rewriting $y = Q^{-1}(x,1-\varepsilon)$ as
\begin{equation}
Q(x,y) = 1-\varepsilon,
\end{equation}
and taking its derivative with respect to $x$ yields
\begin{equation} \label{derQ}
\frac{\partial Q(x,y)}{\partial x} + \frac{\partial Q(x,y)}{\partial y} \frac{\text{d} y}{\text{d} x} = 0.
\end{equation}
By using \eqref{Qprop1-both} in \eqref{derQ} one obtains
\begin{equation}
y\,e^{-\frac{x^2+y^2}{2}}\,I_1(xy) - y\,e^{-\frac{x^2+y^2}{2}}\,I_0(xy) \frac{\text{d} y}{\text{d} x} = 0
\end{equation}
or equivalently
\begin{equation}\label{dy/dx}
\frac{\text{d} y}{\text{d} x} = \frac{I_1(xy)}{I_0(xy)}.
\end{equation}
For small $x$ we have \cite{segura2011bounds}
\begin{equation}
I_n(xy) \cong \left(\frac{xy}{2}\right)^n;~n = 0,1.
\end{equation}
Thus \eqref{dy/dx} can be rewritten as
\begin{equation}
\frac{\text{d} y}{\text{d} x} = \frac{xy}{2},
\end{equation}
which is a first order differential equation with the solution of
\begin{equation}\label{y=y0exp}
y = y_0\,e^{\frac{x^2}{4}},
\end{equation}
where $y_0$ is the value of $y$ at $x = 0$. In order to find $y_0$ we notice that
\begin{equation}\label{Q0}
Q(0,y_0) = 1-\varepsilon,
\end{equation}
and from \cite{short2012computation} one sees that
\begin{equation}\label{Q0-2}
Q(0,y_0) = e^{-\frac{y_0^2}{2}}.
\end{equation}
Thus, by using \eqref{Q0} and \eqref{Q0-2} we have
\begin{equation} \label{y0}
y_0 = \sqrt{-2\ln(1-\varepsilon)}.
\end{equation}
For the large values of $x$ from \eqref{I/I} one can write
\begin{equation}
\frac{I_1(xy)}{I_0(xy)} \cong 1-\frac{1}{2xy},
\end{equation}
which can be used in \eqref{dy/dx} to yield
\begin{equation} \label{y'=1-1/2xy}
\frac{\text{d} y}{\text{d} x} \cong 1-\frac{1}{2xy}.
\end{equation}
In order to solve the above differential equation first we solve the equation by neglecting $1/2xy$. To this end, we rewrite it as
\begin{equation}
\frac{\text{d} y}{\text{d} x} \cong 1.
\end{equation}
This equation has a simple solution as 
\begin{equation}\label{y=x+et1}
y \cong x+\eta_{1\varepsilon},
\end{equation}
where $\eta_{1\varepsilon}$ is a constant determined by $\varepsilon$. Now by using  $y = x+\eta_{1\varepsilon}$, the equation \eqref{y'=1-1/2xy} can be rewritten as
\begin{equation} \label{dy/dx=}
\frac{\text{d} y}{\text{d} x} \cong 1-\frac{1}{2x(x+\eta_{1\varepsilon})} = 1-\frac{1}{2\eta_{1\varepsilon}}\left[\frac{1}{x}-\frac{1}{x+\eta_{1\varepsilon}}\right].
\end{equation}
Therefore, taking the integral of the above equation obtains
\begin{subequations}
\begin{align} 
y &\cong x-\frac{1}{2\eta_{1\varepsilon}}[\ln(x)-\ln(x+\eta_{1\varepsilon})]+\eta_{2\varepsilon} \\ \label{y=x+et2}
&= x-\frac{1}{2\eta_{1\varepsilon}}\ln\left[\frac{x}{x+\eta_{1\varepsilon}}\right]+\eta_{2\varepsilon}.
\end{align}
\end{subequations}
It is to be noted that as $x \rightarrow \infty$, from \eqref{y=x+et2} one finds $y = x + \eta_{2\varepsilon}$ where in comparison with \eqref{y=x+et1} results in
\begin{equation}
\eta_{2\varepsilon} = \eta_{1\varepsilon} \triangleq \eta_\varepsilon.
\end{equation}
In conclusion \eqref{y=x+et2} can be rewritten as
\begin{equation} \label{y(x)}
y \cong x-\frac{1}{2\eta_\varepsilon}\ln\left[\frac{x}{x+\eta_\varepsilon}\right]+\eta_\varepsilon.
\end{equation}
To obtain $\eta_\varepsilon$ one can find that $x \rightarrow \infty$ leads to $y = x + \eta_\varepsilon$ which means that $y \rightarrow \infty$ and $y \gg y-x$. On the other hand, from \cite[2.3--39]{proakis} the conditions $y \rightarrow \infty$ and $y \gg y-x$ results in $Q(x,y) = Q(y-x)$. Thus, since $1-\varepsilon = Q(x,y)$ we have $1-\varepsilon = Q(y-x) = 1-Q(x-y)$ or $Q(x-y) = \varepsilon$. Therefore $Q(-\eta_\varepsilon) = \varepsilon$ or $\eta_\varepsilon = -Q^{-1}(\varepsilon)$.

Note that if $Q^{-1}(\varepsilon) = 0$ the relation in \eqref{y(x)} is ambiguous. To resolve this issue, we re-compute $y$ from \eqref{dy/dx=} by replacing $\eta_{1\varepsilon} = 0$. Thus, we have
\begin{equation}
\frac{\text{d} y}{\text{d} x} \cong 1-\frac{1}{2x^2},
\end{equation}
which results in
\begin{equation} \label{y3}
y \cong x+\frac{1}{2x}+\eta_\varepsilon,
\end{equation}
where $\eta_\varepsilon = -Q^{-1}(\varepsilon)$ as explained before.

In conclusion, from \eqref{y=y0exp}, \eqref{y0}, \eqref{y(x)} and \eqref{y3} for $Q^{-1}(\varepsilon) \neq 0$ we obtain
\begin{equation}
y = 
     \begin{cases}
       \sqrt{-2\ln(1-\varepsilon)}~e^{\frac{x^2}{4}} & ;~x \leq x_0 \\
       x-\frac{1}{2\eta_\varepsilon}\ln\left(\frac{x}{x+\eta_\varepsilon}\right)+\eta_\varepsilon & ;~x > x_0~\wedge~\eta_\varepsilon = -Q^{-1}(\varepsilon)
     \end{cases}
\end{equation}

and for $Q^{-1}(\varepsilon) = 0$ we have
\begin{equation}\label{yfinal2}
y = 
     \begin{cases}
       \sqrt{-2\ln(1-\varepsilon)}~e^{\frac{x^2}{4}} & ;~x \leq x_0 \\
       x+\frac{1}{2x} & ;~x > x_0
     \end{cases}
\end{equation}
where $x_0$ can be determined by the intersection of the sub-functions at $x > \max[0,Q^{-1}(\varepsilon)]$. To clarify this, we note that $y = \sqrt{-2\ln(1-\varepsilon)}~e^{{x^2}/{4}}$ is an strictly increasing function which goes away from $Q^{-1}(x,1-\varepsilon)$ to $+\infty$ at $x \geqslant \max[0,Q^{-1}(\varepsilon)]$, whereas $y = x-\frac{1}{2\eta_\varepsilon}\ln\left(\frac{x}{x+\eta_\varepsilon}\right)+\eta_\varepsilon$ sharply deceases from $+\infty$ to $Q^{-1}(x,1-\varepsilon)$ since $\ln\left[{x}/(x+\eta_\varepsilon)\right]$ is dominant term at the vicinity of  $\max[0,Q^{-1}(\varepsilon)]$. Considering this fact, at a unique $x_0$ these two sub-functions meet each other which is considered as the decision point to switch the approximate values for $y = Q^{-1}(x,1-\varepsilon)$ based on the proposed piecewise function. According to the above discussion, at $x = x_0$ the piecewise function returns the least accurate approximation. The same argument is used for finding the value of $x_0$ in \eqref{yfinal2}.

\begin{figure}[t!]
  \centering
  \includegraphics[width=0.6\columnwidth]{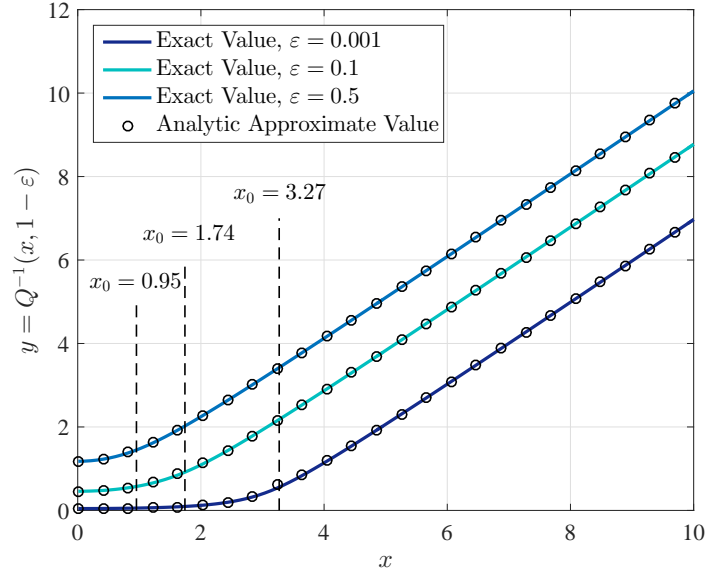}\\
  \caption{The proposed analytic solution for $y = Q^{-1}(x,1-\varepsilon)$ is a good approximation of the exact value.}\label{InvMQfunc_Exact_Approx}
\end{figure}

Figure \ref{InvMQfunc_Exact_Approx} compares the proposed analytic solution for $y = Q^{-1}(x,1-\varepsilon)$ with the exact values. As can be seen, the analytic solution is a good approximate of the exact values.


\section{Proof of Theorem \ref{OptEleAng theorem}} \label{OptEleAng theorem proof}

The derivative of the coverage radius $\rdc$ at its maximum should be zero. Thus 
\begin{equation}
\frac{\partial}{\partial \tC}\rdc = 0
\end{equation}
or equivalently
\begin{equation}\label{Drdp}
\frac{\partial}{\partial \tC}\ln(\rdc) = 0.
\end{equation}
From \eqref{rC} one obtains
\begin{equation}
\ln(\rdc) = \ln \left[\Lambda(\tC)\right] + \ln \left[\cos(\tC)\right],
\end{equation}
and hence by taking the derivative 
\begin{subequations}
\begin{align}
\frac{\partial}{\partial \tC}\ln(\rdc) &= \frac{\partial}{\partial \tC}\ln \left[\Lambda(\tC)\right] + \frac{\partial}{\partial \tC}\ln \left[\cos(\tC)\right] = \frac{\partial}{\partial \tC}\ln \left[\Lambda(\tC)\right] - \tan(\tC).
\end{align}
\end{subequations}
Replacing \eqref{rC} into \eqref{Drdp} yields
\begin{equation}\label{tan=}
\tan(\tC) = \frac{\partial}{\partial \tC}\ln \left[\Lambda(\tC)\right].
\end{equation}
Assuming $\xc \gg \sqrt{2}$ at $\tC = \tCoptdc$, we can simplify $\Lambda(\tC)$ in \eqref{Lambda} as
\begin{equation}\label{LambdaApp}
\Lambda(\tC) \cong \left(\frac{\gamma_\text{U}\,\yc^2}{\xi\,\xc^2}\right)^{\frac{1}{\a(\tC)}},
\end{equation}
and $\yc$ in \eqref{x,y,tC} as 
\begin{equation} \label{yc=xc-}
\yc \cong \xc-Q^{-1}(\varepsilon).
\end{equation}
Using \eqref{tan=} -- \eqref{yc=xc-} one can write
\begin{subequations}
\begin{align}
\tan(\tC) &= \frac{\partial}{\partial \tC} \left[\frac{\ln(\gamma_\text{U}/\xi)+2\ln(\yc)-2\ln(\xc)}{\a(\tC)}\right] \\ \label{eq2}
&= \frac{\partial}{\partial \tC} \left[\frac{\ln(\gamma_\text{U}/\xi)+2\ln[\xc-Q^{-1}(\varepsilon)]-2\ln(\xc)}{\a(\tC)}\right] \\
\label{eq3} &= \frac{2\xc'\left(\frac{1}{\xc-Q^{-1}(\varepsilon)}-\frac{1}{\xc}\right)\a(\tC)-\a'(\tC)\a(\tC)\ln[\Lambda(\tC)]}{\a(\tC)^2} \\
\label{eq4} &= \frac{2\xc'\left(\frac{Q^{-1}(\varepsilon)}{\xc[\xc-Q^{-1}(\varepsilon)]}\right)-\a'(\tC)\ln[\Lambda(\tC)]}{\a(\tC)},
\end{align} 
\end{subequations}
where in \eqref{eq2} and \eqref{eq3} the relations \eqref{yc=xc-} and \eqref{LambdaApp} are used respectively. The equation \eqref{eq4} can be rewritten as
\begin{equation}
\a(\tC)\tan(\tC)+\a'(\tC)\ln[\Lambda(\tC)] = 2\xc'\left(\frac{Q^{-1}(\varepsilon)}{\xc[\xc-Q^{-1}(\varepsilon)]}\right),
\end{equation}
which is the desired result.


\section{Proof of Theorem \ref{theorem - relaying outage}} \label{theorem - relaying outage proof}
Using the total probability theorem, \eqref{out-def-R} can be written as
\begin{equation}\label{Pout=sigma}
\Poutrc = \sum_{i=0}^{\infty}\mathds{P}\left(\Gamma _{\text{R}_J\text{D}}\leq\xi~|~|\A|=i\right)\mathds{P}(|\A|=i) 
\end{equation} 
where $|\A|$ indicates the cardinality of $\A$. Using Marking Theorem \cite{kingman1993poisson}, $\A$ follows PPP with the density obtained as
\begin{subequations}
\begin{align} \nonumber
\lambda_{\A}(\rR,\fiR,h) &= \lambda\,\mathds{P}(\text{R}\in\A) = \lambda\,\mathds{P}(\Gamma_{\text{U}{\text{R}}}>\xi)  \\ 
&\label{lam-A-1} = \lambda\,\mathds{P} \left(\frac{A \PU}{{N_0}\lUR^{\,\a(\tR)}}~\Omega_{\text{U}{\text{R}}}>\xi\right) \\ 
&\label{lam-A-3} = \lambda\,Q\left(\sqrt{2K(\tR)},\sqrt{{2[K(\tR)+1]\,\xi {\lUR^{\,\a(\tR)}}}/{\gamma_\text{U}}}\right)~;~~~\tR = \tan^{-1}(h/\rR),
\end{align}
\end{subequations}
where R is a relay node at an arbitrary location indicated by $(\rR,\fiR)$ in polar coordinates, \eqref{lam-A-1} is obtained using \eqref{GUX} by replacing X with R, and \eqref{lam-A-3} follows from the fact that $\Omega_{\text{U}{\text{R}}}$ has a non-central chi-square PDF expressed in \eqref{OUX} with unit mean. Therefore, $|\A|$ is a Poisson random variable with mean $\mu_{\A}$ computed as
\begin{align} \nonumber
\mu_{\A}(h) &= \int_{\C} \lambda_{\A}(\rR,\fiR,h)\,d\C = \lambda \int_0^{2\pi}\!\!\! \int_0^{\rrc} Q\left(\sqrt{2K(\tR)},\sqrt{{2[K(\tR)+1]\,\xi {\lUR^{\,\a(\tR)}}}/{\gamma_\text{U}}}\right)~\rR\,d\rR\,d\fiR \\
&\label{m-A-1} = 2\pi\lambda \int_0^{\rrc} \rR\,Q\left(\sqrt{2K(\tR)},\sqrt{{2[K(\tR)+1]\,\xi {\lUR^{\,\a(\tR)}}}/{\gamma_\text{U}}}\right)~d\rR,
\end{align}
where $d\C$ is the surface element and \eqref{m-A-1} follows from the fact that $K(\tR) = K\left(\tan^{-1}(h/\rR)\right)$ and $\lUR = \sqrt{h^2+r_\text{R}^2}$ adopt the same value in any $\fiR$. Therefore, the probability mass function of $|\A|$ can be expressed as
\begin{equation}\label{P(A=k)}
\mathds{P}(|\A|=i) = \frac{\mu_{\A}(h)^i}{i!}~e^{-\mu_{\A}(h)}.
\end{equation}
As a result of the PPP property, the locations of the relay nodes in $\A$, i.e. $\text{R}_1,\text{R}_2,...,\text{R}_i$ are identical independent (i.i.d) RVs indicated by R. Thus following the assumption that the fading powers between any pair of nodes are independent RVs, $\Gamma _{\text{R}_1\text{D}},\Gamma _{\text{R}_2\text{D}},...,\Gamma _{\text{R}_i\text{D}}$ become i.i.d RVs as well which are indicated by $\Gamma _{\text{R}\text{D}}$. Therefore, the conditional probability term in \eqref{Pout=sigma} can be calculated as
\begin{subequations}
\begin{align}
\mathds{P}\left(\Gamma _{\text{R}_J\text{D}}\leq\xi~|~|\A|=i\right) &= \mathds{P}\left(\max\{\Gamma _{\text{R}_1\text{D}},\Gamma _{\text{R}_2\text{D}},...,\Gamma _{\text{R}_i\text{D}}\}\leq\xi\right)  \\ \label{cond_Pr}
&= \prod_{j=1}^{i} \mathds{P}\left(\Gamma _{\text{R}_j\text{D}}\leq\xi\right) = \mathds{P}\left(\Gamma _{\text{R}\text{D}}\leq\xi\right)^i.
\end{align}
\end{subequations}
By representing any specific location of $\text{R}$ as $\text{R}: (\rR,\fiR)$ one can write 
\begin{subequations}\label{Gr1d<xi}
\begin{align}
\mathds{P}\left(\Gamma _{\text{R}\text{D}}\leq\xi\right) &= \int_{\C} \mathds{P}\left(\Gamma _{\text{R}\text{D}}\leq\xi~|~\text{R}: (r_\text{R},\varphi_\text{R})\right) \frac{\lambda_{\A}(r_\text{R},\varphi_\text{R},h)}{\mu_{\A}(h)} ~d\C \\ \label{Gr1d<xi-1}
&= \frac{1}{\mu_{\A}(h)}\int_{\C} \mathds{P}\left(\frac{A P_{\text{R}}}{N_0{\lRD}^{\!\!\a_0}}~\Omega_{{\text{R}}{\text{D}}}\leq\xi~|~\text{R}: (r_\text{R},\varphi_\text{R})\right) \lambda_{\A}(r_\text{R},\varphi_\text{R},h) ~d\C \\ \label{Gr1d<xi-2}
&= \frac{1}{\mu_{\A}(h)} \int_{\C} \left[1-Q\left(\sqrt{2\kappa_0},\sqrt{{2(\kappa_0+1)~\xi
{\lRD}^{\a_0}}/{\gamma_\text{R}}}\right)\right] \lambda_{\A}(r_\text{R},\varphi_\text{R},h) ~d\C \\ \nonumber
&= \frac{\lambda}{\mu_{\A}(h)} \int_0^{2\pi}\!\!\!\int_0^{\rrc} r_\text{R}\left[1-Q\left(\sqrt{2\kappa_0},\sqrt{{2(\kappa_0+1)~\xi
{\lRD}^{\a_0}}/{\gamma_\text{R}}}\right)\right] \\ \label{Gr1d<xi-4}
&\times Q\left(\sqrt{2K(\tR)},\sqrt{{2[K(\tR)+1]~\xi {{\lUR}^{\a(\tR)}}}/{\gamma_\text{U}}}\right) dr_\text{R}d\varphi_\text{R}.
\end{align}
\end{subequations}
In \eqref{Gr1d<xi-1} the expression in \eqref{G_G1G2} is used, \eqref{Gr1d<xi-2} follows from the fact that $\Omega_{{\text{R}}{\text{D}}}$ has a non-central chi-square PDF with unit mean where $\gamma_\text{R} = A P_\text{R}/N_0$, and in \eqref{Gr1d<xi-4} the relation in \eqref{lam-A-3} is replaced.

Now from \eqref{Pout=sigma}, \eqref{P(A=k)} and \eqref{cond_Pr} one can write
\begin{subequations}
\begin{align}
\Poutrc(\rD,h) &= \sum_{i=0}^{\infty} \mathds{P}\left(\Gamma _{\text{R}\text{D}}\leq\xi\right)^i \frac{\mu_{\A}(h)^i}{i!}~e^{-\mu_{\A}(h)} = e^{-\mu_{\A}(h)} \sum_{i=0}^{\infty} \frac{\left[\mu_{\A}(h) \mathds{P}\left(\Gamma _{\text{R}\text{D}}\leq\xi\right)\right]^i}{i!} \\ \label{Pout-final-2}
&= e^{-\mu_{\A}(h)} e^{\mu_{\A}(h) \mathds{P}\left(\Gamma _{\text{R}\text{D}}\leq\xi\right)} = e^{-\mu_{\A}(h)+\mu_{\A}(h) \mathds{P}\left(\Gamma _{\text{R}\text{D}}\leq\xi\right)}
\end{align} 
\end{subequations}
where \eqref{Pout-final-2} is obtained using the Taylor series expansion of exponential function. From \eqref{m-A-1}, \eqref{Gr1d<xi-4} and \eqref{Pout-final-2} we obtain the desired result.


\section{Path Loss Exponent $\a(\theta)$} \label{PathLossExp}

Here we discuss the proposed expression for $\a(\theta)$ in \eqref{alfa-Plos} by using a model reported in \cite{al2014optimal} in which the path loss can be written as
\begin{subequations}\label{PL1}
\begin{equation}
\text{PL}_1(\theta,\ell) = \text{PL}_\text{LoS}(\ell) \cdot \PLoS(\theta) + \text{PL}_\text{NLoS}(\ell) \cdot [1-\PLoS(\theta)],
\end{equation}
where
\begin{align}
\text{PL}_\text{LoS}(\ell) &= 20\log\left(\frac{4 \pi f}{c} \ell \right) + \sigma_\text{LoS}, \\ \label{PLNLoS}
\text{PL}_\text{NLoS}(\ell) &= 20\log\left(\frac{4 \pi f}{c} \ell \right) + \sigma_\text{NLoS},
\end{align}
\end{subequations}
$\PLoS(\theta)$ is given in \eqref{PrLoS}, $f$ is the system frequency, $c$ is the speed of light, $\ell$ is the distance between transmitter and receiver, and $\sigma_\text{LoS}$ and $\sigma_\text{NLoS}$ are excessive path loss corresponded to the LoS and NLoS signals respectively which are constants being independent of $\theta$. On the other hand, in our model presented in Section II the path loss in dB is
\begin{equation}\label{PL2}
\text{PL}_2(\theta,\ell) = 10\a(\theta)\log(\ell) + A_\text{dB},
\end{equation} 
where $A_\text{dB} = 10\log(A)$. Thus in order to fit the two models one can write
\begin{equation}\label{alfa-mod1}
\a(\theta) = \frac{\text{PL}_1(\theta,\ell) - A_\text{dB}}{10\log(\ell)}.
\end{equation} 
However the above equation results in a distance-dependent $\a(\theta)$. To resolve this issue we take the average of $\a(\theta)$ obtained from \eqref{alfa-mod1} over a reasonable range of communication, $\ell_1,\ell_2,...,\ell_\text{N}$, for a UAV. Therefore, one obtains
\begin{equation}
\a(\theta) = \frac{1}{N} \sum_{i = 1}^{N} \frac{\text{PL}_1(\theta,\ell_i) - A_\text{dB}}{10\log(\ell_i)},
\end{equation}
which can be rewritten as
\begin{equation}
\a(\theta) = a_1 \cdot \mathcal{P}_\text{LoS}(\theta) + a_2,
\end{equation} 
where $a_1$ and $a_2$ are determined by the type of environment (suburban, urban, dense urban, ...) and system frequency, and given by
\begin{align}
a_1 &= \sum_{i = 1}^{N} \frac{\text{PL}_\text{LoS}(\ell_i)-\text{PL}_\text{NLoS}(\ell_i)}{10N\log(\ell_i)}, \nonumber \\
a_2 &= \sum_{i = 1}^{N} \frac{\text{PL}_\text{NLoS}(\ell_i)-A_\text{dB}}{10N\log(\ell_i)}.
\end{align}

It is to be noted that $\text{PL}_1(\theta,\ell)$ limits the model for large $\theta$s where free space assumption is met. To clarify this fact, assume that $\theta_0$ is a very low angle where $\PLoS(\theta_0) \rightarrow 0$. Thus, using \eqref{PL1} one obtains
\begin{equation}
\text{PL}_1(\theta_0,\ell) \cong \text{PL}_\text{NLoS}(\ell) = 20\log\left(\frac{4 \pi f}{c} \ell \right) + \sigma_\text{NLoS},
\end{equation}
where $\sigma_\text{NLoS}$ is a constant parameter independent from distance $\ell$. Therefore, the above equation is not following the well-known path loss behavior of the G2G communication where the path loss exponent is expected to be larger than 3. However, our proposed model is capable of resolving this issue by setting an appropriate value for $\a_0$ in a way that satisfies the G2G communication propagation characteristics for low angles.


\bibliographystyle{IEEEtran}
\bibliography{Refs}

\end{document}